\documentclass[a4paper]{article}

\usepackage{a4wide}
\usepackage[T1]{fontenc}
\usepackage{amsmath,amssymb,amsthm}
\usepackage{graphicx}
\usepackage{enumitem}
\usepackage{tcolorbox}
\usepackage{hyperref}

\usepackage{mathrsfs}
\usepackage{cmll}
\usepackage{etex}
\usepackage{amsmath} 
\usepackage{amssymb}
\usepackage{amsfonts}
\usepackage{mathtools} 
\usepackage{stmaryrd}
\SetSymbolFont{stmry}{bold}{U}{stmry}{b}{n}
\usepackage{mathpartir}
\usepackage{xargs}
\makeatletter
\define@key {mprset}{style}[1]{\def\TirNameStyle{#1}}
\makeatother
\mprset{style={\scriptsize\scshape}}
\usepackage{etoolbox}
\usepackage{tikz,tikz-cd}
\usetikzlibrary{arrows,arrows.meta,positioning}
\usepackage{url}
\usepackage{subcaption}
\usepackage{graphicx}
\usepackage{caption}
\usepackage[export]{adjustbox}

\usepackage{thmtools}

\usepackage{multirow} 
\usepackage{makecell} 
\newcounter{number}
\newenvironment{varitemize}
{
	\begin{list}{\labelitemii}
		{\setlength{\itemsep}{0.0mm}
			\setlength{\topsep}{0.0mm}
			\setlength{\parindent}{0.0mm}
			\setlength{\parskip}{0.0mm}
			\setlength{\parsep}{0.0mm}
			\addtolength{\leftmargin}{-\labelsep}}}
	{
	\end{list}
}

\usepackage[draft]{optional}
\usepackage{comment}
\usepackage[colorinlistoftodos,prependcaption,textsize=tiny]{todonotes}
\usepackage{xcolor}
\usepackage[all]{xy}
\usepackage{braket}
\usetikzlibrary{fit}

\usepackage{graphicx}
\graphicspath{  }

\makeatletter
\newcommand{\oset}[3][0ex]{%
	\mathrel{\mathop{#3}\limits^{
			\vbox to#1{\kern-2\ex@
				\hbox{$\scriptstyle#2$}\vss}}}}
\makeatother

\newcommand{\lis}[1]{\langle #1 \rangle} 

\newcommand{\den}[1]{\llbracket #1 \rrbracket} 

\newcommand{\eval}{\mathbf{eval}} 

\newcommand{\dlambda}{\boldsymbol{\lambda}}

\newcommand{\id}{\mathrm{id}} 

\newcommand{\ISet}{\mathbf{ISet}}
\renewcommand{\Set}{\mathbf{Set}}

\newcommand{\Store}{\mathrm{St}}

\newcommand{\supp}[1]{\mathbf{supp}{(#1)}}

\newcommand{\BB}{\mathbb{B}}
\newcommand{\Str}{\mathbb{S}}

\newcommand{\UU}{\mathbb{U}}

\newcommand{\N}{\mathbb{N}}

\newcommand{\lamBLL}{{\lambda}\mathbf{BLL}}

\newcommand{\Val}{\mathcal{V}}

\newcommand{\letin}[2]{\mathtt{let}\,#1\,\mathtt{in}\, #2}
\newcommand{\ifthen}[3]{\mathtt{if}\,#1\,\mathtt{then}\, #2\,\mathtt{else}\; #3}
\newcommand{\loopp}[3]{\mathtt{loop}\,#1\,#2\,\mathtt{times \, from}\, #3}
\newcommand{\setr}[2]{\mathtt{set}\,#1\, #2}
\newcommand{\return}{\mathtt{return}\,}
\newcommand{\der}{\mathtt{der}}
\newcommand{\get}{\mathtt{get}\,}

\newcommand{\true}{\mathbf{t}}
\newcommand{\false}{\mathbf{f}}

\newcommand{\typeof}[1]{\mathsf{typeof}(#1)}
\newcommand{\Fun}{\mathcal{F}}
\newcommand{\algof}{\mathsf{alg}}
\newcommand{\random}{\mathtt{random}}
\newcommand{\equals}{\mathit{equal}}
\newcommand{\xor}{\mathtt{xor}}

\newcommand{\flipcoin}{\mathit{flipcoin}}

\newcommand{\subst}[1]{#1}

\newcommand{\vdashc}[1][]{\vdash^{#1}_{c}}
\newcommand{\vdashv}[1][]{\vdash^{#1}_{v}}

\newcommand*{\Comp}[4]{\Lambda_{#1}(#2;#3; #4)}
\newcommand*{\CompC}[3]{\Lambda_{#1}(#2;#3)}
\newcommand*{\Values}[4]{\V_{#1}(#2;#3; #4)}
\newcommand*{\ValuesC}[3]{\V_{#1}(#2;#3)}

\newcommand*{\CompO}[3]{\Lambda_o(#1 #2; #3)}

\newcommand*{\ValuesO}[3]{\V_o(#1 #2; #3)}
	
\newcommandx*{\smstep}[3][1={}, 2={},3={}]{\mathbf{sm}^{#1;#2}_{#3}}

\newcommandx*{\itsmstep}[3][1={}, 2={},3={}]{\mathbf{ism}^{#1;#2}_{#3}}

\newcommand{\cost}{\mathbf{cost}}

\newcommand{\BPP}{\ensuremath{\mathsf{BPP}}}
\newcommand{\ZPP}{\ensuremath{\mathsf{ZPP}}}

\newcommand{\D}{\mathbf{D}} 
\newcommand{\Dist}{\mathscr{D}} 
\newcommand{\T}{\mathbf{T}}

\newcommand{\V}{\mathcal{V}}

\newcommand{\ie}{\emph{i.e.}\,}

\newcommand{\bindsymbol}{\scalebox{0.5}[1]{$>\!>=$}}
\newcommand{\bind}{\mathrel{\bindsymbol}}

\usetikzlibrary{decorations.markings}
\tikzset{module/.style={%
		decoration={markings,%
			mark= at position 0.5 with {%
				\node[] (tempnode) {};
				\draw[-,shorten <=.1em,shorten >=.1em,inner sep=1pt]
				(tempnode.north) -- (tempnode.south);
			}
		},
		postaction={decorate}
	}
}

\def\typstr#1{\mathbb{S}[#1]}
\def\typbool{\mathbb{B}}

\def\Gen{\mathit{Gen}}
\def\Enc{\mathit{Enc}}

\def\typmsgprivk{\typstr{i}}

\def\typkeyprivk{\typstr{i}}
\def\typunit{\mathbb{U}}
\def\lollipop{\multimap}
\def\tensor{\otimes}

\def\unusedvar{-}

\def\emptycontextCPA{\varnothing}

\def\FV#1{\mathit{FV}(#1)}

\newcommand{\giulia}[1]{}
\opt{draft}{
   \renewcommand{\giulia}[1]{\todo[color=green!30]{#1 -- G} \PackageWarning{TODO}{#1 -- G}}
}

\newcommand{\ugo}[1]{}
\opt{draft}{
   \renewcommand{\ugo}[1]{\todo[color=blue!40]{#1} \PackageWarning{TODO}{#1}}
}

\def\polyInferRule#1{\mathit{poly}\left(\raisebox{-0.7em}{\ensuremath{#1}}\right)}
\def\polyStarInferRule#1{\mathit{\underline{poly}}\left(\raisebox{-0.7em}{\ensuremath{#1}}\right)}

\def\poly#1{\mathit{poly}(#1)}
\def\polyStar#1{\mathit{\underline{poly}}(#1)}

\def\sizeM#1{\mathbf{size}(#1)}
\def\sizeV#1{\mathbf{size}(#1)}

\def\sizeZ#1{\mathbf{size}(#1)}

\def\sizeG#1{\mathbf{size}(#1)}
\def\sizeP#1{\mathbf{size}(#1)}

\def\enc#1{\ulcorner #1\urcorner}

\def\tapeAlph{\Psi}
\def\blank{\square}

\def\ConfSet{\underline{\mathit{Conf}}} 
\def\typTape{\mathit{Tape}}

\def\typConf{\mathit{Conf}}
\def\typState{\mathit{State}}
\def\typSymb{\mathit{Symbol}}
\def\typDir{\mathit{Direction}}

\def\transitionTerm#1{\mathit{Transition}_{#1}}
\def\writeSymbTerm{\mathit{WriteSymb}}
\def\moveLTerm{\mathit{MoveLeft}}
\def\moveRTerm{\mathit{MoveRight}}
\def\evalTerm{\mathit{Eval}}
\def\moveHeadTerm{\mathit{MoveHead}}
\def\stepPTMTerm{\mathit{Step}}
\def\encodeInitialStateTerm{\mathit{EncodeInitialState}}
\def\fillBlankStringTerm{\mathit{FillBlankString}}
\def\fillBlankTapeTerm{\mathit{FillBlankTape}}
\def\encodeInputStringTerm{\mathit{EncodeInput}}
\def\initInputTapeTerm{\mathit{InitInputTape}}
\def\initTerm{\mathit{Init}}
\def\decodeOutputStringTerm{\mathit{DecodeOutput}}
\def\finalTerm{\mathit{Final}}

\def\polyInferRule#1{\mathit{poly}\left(\raisebox{-0.7em}{\ensuremath{#1}}\right)}

\def\poly#1{\mathit{poly}(#1)}

\newenvironment{sproof}{%
	\begin{proof}%
	}{%
	\end{proof}%
}

\usepackage{amsmath}
\usepackage{amsthm}
\usepackage{amssymb}

\theoremstyle{plain}
\newtheorem{theorem}{Theorem}[section]
\newtheorem{lemma}[theorem]{Lemma}
\newtheorem{corollary}[theorem]{Corollary}
\newtheorem{proposition}[theorem]{Proposition}
\theoremstyle{definition}
\newtheorem{definition}[theorem]{Definition}

\title{(Implicitly) Characterizing Probabilistic Polynomial Time
	as Used in Cryptography\thanks{%
		The first author is partially supported by the ANR Project HOPR
		(ANR-24-CE48-5521-01).}}

\author{%
	Ugo Dal Lago\footnote{University of Bologna, Italy, \texttt{ugo.dallago@unibo.it}}\\
	\and
	Giulia Giusti\footnote{Independent Researcher, Italy, \texttt{giustigiulia97@gmail.com}}}

\date{}

\begin{document}
	
\maketitle

\newcommand{\LambdaBLL}{\ensuremath{\lambda\textbf{BLL}}}

\begin{abstract} 
In this paper, we show how a recently introduced $\lambda$-calculus for 
randomized computation, called $\LambdaBLL$, can accurately capture the 
notions of first- and second-order polytime computations used in cryptographic 
definitions and proofs. These results, spelled out in the style of so-called \emph{implicit 
computational complexity}, imply that $\LambdaBLL$ is sufficiently powerful to express 
experiments and adversaries as used in cryptography, and that reductions, when 
formulated within the calculus, are guaranteed to be complexity-preserving. Remarkably, 
while polynomial-time \emph{soundness} holds for the full calculus, polynomial-time 
\emph{completeness} already follows in presence of a minimal finite set of basic functions. 
\end{abstract}

\section{Introduction}
\begin{sloppypar}
	Implicit computational complexity (ICC) aims to characterize complexity classes 
	using logical systems and programming languages, thereby avoiding explicit 
	reference to machines and refraining from specifying resource bounds 
	directly~\cite{BellantoniCook92,LeivantMarion93,Hofmann1999,DalLago2011}. 
	These two aspects, central in the standard definition of most complexity classes,
	make such definitions less suitable for program verification. The 
	counterparts to complexity classes provided by ICC have in fact motivated the 
	development of several verification methodologies for reasoning about program 
	complexity, such as amortized analysis~\cite{HoffmannHofmann2010,HAH2012,HLOMZ2022} 
	and restricted forms of path orders~\cite{BonfanteMarionMoyen2011,AvanziniMoser2013}. Naturally, 
	verifying that programs satisfy complexity bounds is a difficult but relevant 
	problem, which becomes particularly critical in those
	contexts where resources are scarce, as is the case in quantum 
	computing~\cite{ColledanDalLago2025}.
\end{sloppypar}

Another domain where complexity plays a fundamental role is the \emph{computational 
model} to cryptography~\cite{GoldwasserMicali,katz2020introduction}. 
In this setting, 
security is spelled out as a form of (approximate) 
resistance to attacks performed by randomized, \emph{complexity-bounded} adversaries. 
Proofs are typically established via reductions, since unconditional results 
are rare. It is therefore essential that reductions preserve complexity: given 
an adversary breaking the claim, one must construct an adversary against the 
assumption \emph{and} prove that the complexity of the latter is \emph{not} substantially 
higher than the complexity of the former. Calculi and tools designed to model 
such proofs must, accordingly, incorporate mechanisms to reason about 
computational complexity. Most tools and languages designed to support the construction of formal 
cryptographic proofs, however, do not address complexity constraints
explicitly~\cite{BGHB2011,BLS2020}, typically assuming 
that the underlying reduction is complexity-preserving, although some attempts 
in this direction have recently appeared in the literature~\cite{BBGKS2023}. Noticeably, the 
issue can be circumvented by appealing to computationally complete
symbolic models~\cite{BC2012}, in which the adversary's 
complexity is, in some sense, taken care of by the underlying logic.

It is natural to ask whether ICC techniques can be applied 
to the characterization of probabilistic polynomial time as used in 
cryptography, thereby ensuring that all programs, particularly those 
representing adversaries and the reductions between them, \emph{are} efficient or 
\emph{preserve} efficiency. The literature offers a number of attempts in this 
direction, to which we shall return in Section \ref{sect:conclusion}. In none of these cases, 
however, has a precise correspondence with polynomial time \emph{as used by 
	cryptographers} been established.

But in what sense does polynomial time as used by cryptographers differ from 
probabilistic polynomial time as traditionally defined in computational 
complexity theory, e.g., via the classes \BPP\ or \ZPP? The main difference lies 
in \emph{what} the complexity class contains, and in how the \emph{input size} is defined. 
For standard complexity classes, which are classes of functions on binary strings,
the aforementioned parameter is simply the number of bits required to represent the input. 
In the cryptographic setting, however, one rather considers the 
so-called \emph{security parameter}: a publicly known natural number that measures 
key length, and with respect to which all exchanged data must have polynomially 
bounded size. This distinction is only mildly visible at first order, but 
becomes substantially more pronounced in the presence of higher-order 
functions, which are pervasive in cryptography, as we will discuss in Section 
\ref{sect:soundcomplete} below.

The goal of this work is precisely to investigate the expressive power of one of 
the calculi supporting cryptographic proofs, namely 
$\lamBLL$~\cite{aplas}. More specifically, we show that this calculus captures 
\emph{random functions} computable in probabilistic polynomial time which, arguably, 
corresponds to what cryptographers implicitly rely on in definitions and 
proofs. This correspondence holds not only at first order but also at second 
order. In the first-order case, the correspondence applies to the calculus as 
introduced in \cite{aplas}, and it is established in Section 
\ref{sect:soundcomplete}. For the second-order setting, some small modifications to 
the formulation are required, and these refinements are discussed in Section 
\ref{sect:secondorder}. The notion of polynomial time that we adopt is 
introduced in Section \ref{sect:soundcomplete}, while the calculus and its 
semantics are presented in Section \ref{sect:calculus}.  
In Section \ref{sec:cryptography}, we briefly explain, without going into technical 
details, how $\lamBLL$ can be helpful when carrying out cryptographic proofs. 
Before delving into the technical details of the calculus, this manuscript 
devotes a few introductory paragraphs in the next section to the concept of efficiency 
typically used in cryptography, motivating the need for introducing new complexity classes.
\section{About Complexity and Cryptography}\label{sect:birdseye}
In this section, we address the following question: \emph{what do the efficient 
algorithms studied in cryptography actually compute?} One could of course give 
the obvious answer: such algorithms compute functions over the set of binary 
strings, i.e., over $\{0,1\}^*$. But is this an accurate description? If we 
examine the structure of cryptographic primitives (e.g., pseudorandom 
generators~\cite{BlumMicali1984}) and of the distinguishers that act upon them, we realize that
what they compute is \emph{a family} of functions $\{f_n\}_{n \in \mathbb{N}}$, where
\(f_n : \{0,1\}^{p(n)} \rightarrow \{0,1\}^{q(n)}\),
and $p$ and $q$ are polynomials. The entire family is required to be 
computable either by a \emph{single} algorithm running in time polynomial in $n$, or by 
a family of Boolean circuits of polynomial size, again in $n$. When the 
algorithm or circuit has access to a source of randomness, one may assume, following the 
definition of \textsf{BPP}, that the error is bounded, or
that every $f_n$ outputs not a \emph{single} element of 
$\{0,1\}^{q(n)}$ but rather a \emph{probability distribution} over that same set, i.e.,
\begin{equation}\label{equ:firstorder}
	f_n : \{0,1\}^{p(n)} \rightarrow \mathbf{D}(\{0,1\}^{q(n)}),
\end{equation}
where $\mathbf{D}(X)$ is the set of distributions with support $X$.
It is clear that this picture does not deviate significantly from the 
mainstream of complexity theory; nevertheless, it already illustrates why most 
ICC characterizations of polynomial time are not directly applicable in this 
setting. The point is precisely that the functions involved are, in the 
terminology of \cite{SvEK2009}, \emph{shapely}.

The first contribution of this paper consists, precisely, in showing that first-order typed terms 
of $\lamBLL$, as introduced in \cite{aplas}, compute all and only 
those (families of) functions of the form (\ref{equ:firstorder}) that are computable 
by a probabilistic Turing machine (PTM in the following) running in time polynomial in $n$. This is 
proved by showing, on the one hand, that terms of the calculus can be 
evaluated in polynomial time, and on the other hand, that every probabilistic 
Turing machine running in polynomial time can be simulated in $\lamBLL$. 
The interest of the completeness result lies in the fact that, while $\lamBLL$ in principle
provides function symbols for \emph{any} function computable in polynomial time, 
the encoding of PTMs requires only a handful such functions, e.g. 
circular shifts, bit reads and bit writes.

But what if we consider \emph{second-order}  
adversaries, such as distinguishers for pseudorandom functions~\cite{GMM1986} or adversaries 
against message authentication codes~\cite{BCK96}? In both cases, such adversaries are 
efficient algorithms with access to an \emph{oracle} computing a first-order function, 
and they return a string or even a boolean value. We could therefore think that 
such adversaries compute functions of the form
\(f: (\{0,1\}^*\rightarrow\{0,1\}^*)\rightarrow \{0,1\}^*\).
Computational complexity offers various definitions of efficiency for \emph{second-order} 
functions~\cite{Mehlhorn1976,KapronCook1996}. 
For instance, it would be natural to require that such a function be 
computed by a Basic Feasible Functional~\cite{KapronCook1996}, or by a higher-order adaptation thereof. 
Again, however, this does not exactly correspond to what cryptographers actually assume. 
In the cases mentioned above, in fact, cryptographers assume that the adversary 
computes a family of functions $f_n$, where each of these functions now has the following 
form:
\begin{equation}\label{equ:secondorder}
	f_n : (\{0,1\}^{p(n)} \rightarrow \mathbf{D}(\{0,1\}^{q(n)}))\rightarrow \mathbf{D}(\{0,1\}^{r(n)}),
\end{equation}
Such a family of functions, indexed by the security parameter $n$, is efficiently computable 
when there exists a probabilistic \emph{oracle} Turing machine running in time polynomial in $n$ which, 
given as input the security parameter represented in unary form, interacts with the oracle 
computing an arbitrary function $g_n$ in the manner prescribed by the polynomials $p$ 
and $q$, and then returns the string $f_n(g_n)$. In Section \ref{sect:secondorder}  
we will describe in what sense $\lamBLL$ also characterizes these functions through second-order terms.

From a technical perspective, then, the contribution of this paper consists of soundness 
and completeness results for $\lamBLL$ with respect to novel classes of polynomial-time random functions
we introduce below and which capture the notion of efficiency
arising in cryptographic practice. Why are these results interesting, and why are they useful when
considering $\lamBLL$ as a formal framework for cryptographic proofs?

\begin{sloppypar}
First, the aforementioned results guarantee that, when $\lamBLL$ is used to express cryptanalytic attacks, 
\emph{any} attack described by a $\lamBLL$ term is guaranteed to run in polynomial time. There is more:
completeness ensures that the language is expressive enough to represent \emph{every} such 
polynomial-time attack.  There is, however, an indirect but ultimately more significant reason. 
Security proofs for cryptographic primitives are typically carried out by reduction: one constructs 
a hypothetical attack against an assumed-secure primitive from an equally hypothetical attack 
against the primitive whose security is to be established. A crucial requirement is that the reduction 
has to be \emph{complexity preserving}. Once again, both soundness and completeness play a fundamental role 
in establishing this property, as we will explain in more detail in Section~\ref{sec:cryptography} below.
\end{sloppypar}
\section{The Syntax and Operational Semantics of $\lamBLL$}\label{sect:calculus} 
The language $\lamBLL$ was introduced in \cite{aplas}, drawing inspiration from Bounded 
Linear Logic~\cite{girard_bounded_1992}, graded $\lambda$-calculi~\cite{reed_distance_2010} 
and Levy's CBPV~\cite{levy2012call,ehrhard2019probabilistic}. We briefly recall its key features here. 

\subparagraph*{Types}
At the type level, $\lamBLL$ adopts a graded linear typing discipline, whereas
grades are given by polynomials and are used to control the computational power of terms, see Figure \ref{fig:syntax}, left. 
These polynomials are formed from positive natural numbers in $\N_{\ge 1}$, addition, and multiplication, 
and also involve a distinguished variable $i$ representing the \emph{security parameter}, enabling 
reasoning over families of types and terms indexed by $i$, as in $\mathbf{BLL}$ \cite{girard_bounded_1992}.  In the following, $\N_{\ge 1}[i]$
stands for the set of all such polynomials. 
\begin{figure}[h]
	\centering 
	\scriptsize
	\begin{tcolorbox}[boxrule=0.3pt, colback=white]
	\begin{align*} 
		&\hspace*{-2mm}\textbf{Ground types}  && G ::= \UU \mid \BB \mid \Str[p] &\textbf{Ground values} \;& W ::= \star \mid \true \mid \false \mid s  \\
		& \hspace*{-2mm} \textbf{Positive types} \; && P ::= G \mid P \otimes P \mid \oc_p A &\textbf{Positive values} \; & Z ::= x \mid W\mid \lis{Z,Z} \mid \oc M \\
		&\hspace*{-2mm} \textbf{Types} \; && A ::= P \mid P \multimap A &\textbf{Values}  \;  &\Val \ni V ::= Z \mid \lambda x . M\\
		& &&	&\textbf{Computations}   \; & \Lambda \ni M ::=  \return V \\
		&\hspace*{-2mm}  \textbf{Variable contexts} && \Gamma ::= \varnothing \mid x :P, \Gamma & & 
		\hspace*{-8mm}\mid  \der (Z) \mid MZ\mid \letin{x=N}{M}
		     \\
		&\hspace*{-2mm}\textbf{Reference contexts} && \Theta ::= \varnothing \mid r : G , \Theta & & \hspace*{-8mm} \mid  f_{p}(Z_1, \dots, Z_m)  \mid 	\setr{r}{Z}  \\  
		& && & &  \hspace*{-8mm} \mid \ifthen{Z}{M}{N} \mid	 \get r \\  
		& \hspace*{-2mm} \textbf{Polynomials} \; && p ::= 1 \mid i \mid p + p \mid p\times p 
		& &  \hspace*{-8mm} \mid \letin{\lis{x,y}=Z}{M}   \\
		& && & & \hspace*{-8mm} \mid  \loopp{V}{p}{M}  \\
	\end{align*}
	\end{tcolorbox}
	\caption{Syntax of $\lamBLL$}
	\label{fig:syntax}
\end{figure}

The base types include the unit type $\UU$, booleans $\BB$, and binary strings 
$\Str[p]$ whose length is determined by a polynomial~$p\in\N_{\ge 1}[i]$. The system distinguishes between positive
types and general types, \emph{à la} CBPV~\cite{levy2012call,ehrhard2019probabilistic}, to restrict 
function arguments to be of \emph{positive} type. In addition to a tensor product and a linear function space, 
the system features a graded modality allowing the construction of a positive type $!_p A$ from an arbitrary type 
$A$. This modality makes it possible, in particular, to construct higher-order types. Observe that polynomials 
over $\mathbb{N}_{\ge 1}[i]$ may occur in types. It is therefore natural to consider the operation of substituting 
every occurrence of the variable $i$ in a type $A$ with a polynomial $p$. Since there is only one variable around, 
we write the resulting type as $Ap$, rather than using the more explicit notation $A[i \leftarrow p]$.

\subparagraph*{Terms}
At the term level, instead, $\lamBLL$ extends a linear $\lambda$-calculus in CBPV style~\cite{ehrhard2019probabilistic}
with two key features, essential for modelling 
cryptographic primitives: \emph{function symbols} representing probabilistic polynomial-time computations over 
strings, and \emph{global references}, which enable the representation of shared mutable state across different 
components of an experiment. The grammar of ground and positive values, together with that of
computations, is in Figure \ref{fig:syntax}, right.

But where do the function symbols come from? 
Formally, we  fix a set of function symbols $\Fun$, each function symbol $f$ in $\Fun$ coming with:
\begin{varitemize}
	\item[$\bullet$] a \emph{type}, denoted $\typeof{f}$, of the form $G_1 \times \dots \times G_m \to G$ where $G_1, \dots, G_m$ and $G$ are ground types;
	\item[$\bullet$] for every polynomial $p$ in  $\N_{\geq 1}[i]$, a \emph{term constructor} $f_{p}$ of arity $m$.
\end{varitemize} 
For example, we have the nullary function symbol $\random$, with $\typeof{\random} : \Str[i]$, and the binary function symbol $\xor$, with $\typeof{\xor} : \Str[i] \times \Str[i] \to \Str[i]$. A precise semantic interpretation of these symbols will be given below. Intuitively, however, we can already say that
$\random$ generates a random binary string when evaluated, whereas $\xor$ computes the bitwise exclusive-or of its two argument strings.

The terms $\setr{r}{V}$ and $\get r$ update and read the value stored in the global reference $r$, enabling shared mutable state across protocol components.
Furthermore, to model experiments where an adversary may query an oracle polynomially many times, $\lamBLL$ incorporates a $\mathtt{loop}$ iterator. 

The set of free variables $\FV{M}$ of a term $M$ is defined by induction as usual, and terms are taken modulo $\alpha$-equivalence.
Moreover, polynomials may also occur in terms, and the notation for substituting the variable $i$ with a polynomial $p$ extends naturally to terms. 
Thus, for example, $Mp$ denotes the computation obtained by replacing every occurrence of $i$ in $M$ with $p$.


\subparagraph*{Typing Rules}
\begin{figure}[h!]
	\vspace*{-4mm}
	\begin{tcolorbox}[boxrule=0.6pt, colback=white]
	\begin{center} \scriptsize
		\vspace*{-5mm}
		\begin{mathpar}\mprset{sep=1em,  vskip =0.05ex}
			\inferrule* [right=var]{\\}{ x : P \vdashv[\Theta] x: P}
			\and 
			\inferrule* [right=true]{\\}{ \vdashv[\Theta] \true : \BB}	
			\and
			\inferrule* [right=false]{\\}{ \vdashv[\Theta] \false : \BB}	
			\and
			\inferrule* [right=fun]{ \mathrm{typeof}(f) = G_1 \times \dots \times G_m \to G \\ (\Gamma_k\subst{p} \vdashv[\Theta] Z_k : G_k\subst{p})_{1\leq k \leq m} \\ p \in \N_{\geq 1}[i]}{ \boxplus_k \Gamma_k\subst{p} \vdashc[\Theta] f_p(Z_1, \dots, Z_m):G\subst{p}}
			\and
			\inferrule* [right=string]{s\in \{0,1\}^{c} \\ p : i \mapsto c\text{ is a constant polynomial}}{ \vdashv[\Theta] s : \Str[p]}	
			\and
			\inferrule* [right=tensor]{\Gamma \vdashv[\Theta] Z_1 : P \\ \Delta  \vdashv[\Theta] Z_2: Q}{\Gamma \boxplus \Delta \vdashv[\Theta] \lis{Z_1, Z_2} : P\otimes Q}	
			\and
			\inferrule* [right=let\_pair]{\Gamma \vdashv[\Theta] Z :  P \otimes Q\\x : P, y : Q, \Delta \vdashc[\Theta]  M : A}{\Gamma \boxplus \Delta \vdashc[\Theta]  \letin{ \lis{x,y} = Z}{M} : A }
			\and
			\inferrule* [right=unit]{\\}{ \vdashv[\Theta] \star : \UU}	
			\and
			\inferrule* [right=bang]{\Gamma \vdashc[\Theta] M : A}{p \ast \Gamma \vdashv[\Theta] \oc M : \oc_p A}	
			\and
			\inferrule* [right=der]{\Gamma \vdashv[\Theta] Z : \oc_1 A }{\Gamma \vdashc[\Theta] \der (Z):A } 	
			\and
			\inferrule* [right=app]{\Gamma  \vdashc[\Theta] M : P\multimap A
				\\ \Delta \vdashv[\Theta] Z : P}{\Gamma\boxplus \Delta  \vdashc[\Theta] MZ : A}	
			\and	\inferrule* [right=lam]{\Gamma, x: P  \vdashc[\Theta] M : A }{\Gamma \vdashv[\Theta] \lambda x . M : P\multimap A}
			\and
			\inferrule* [right=eta]{\Gamma  \vdashv[\Theta] V: A}{\Gamma  \vdashc[\Theta] \return V: A}	
			\and
			\inferrule* [right=let]{\Gamma  \vdashc[\Theta] N: P \\ x : P, \Delta  \vdashc[\Theta] M : A }{ \Gamma\boxplus \Delta  \vdashc[\Theta] \letin{ x = N}{M} : A}	
			\and
			\inferrule* [right=loop]{\Gamma \vdashv[\Theta] V: P\multimap P \\ \Delta  \vdashc[\Theta]  M : P }{(p \ast \Gamma)\boxplus \Delta  \vdashc[\Theta]  \loopp{V}{p}{M} : P }	
			\and
			\inferrule* [right=set]{  \Gamma \vdashv[\Theta] Z: G \\  r :G \in \Theta}{ \Gamma \vdashc[\Theta] \setr{r}{Z} : \UU }	
			\and 
			\inferrule* [right=get]{r :G \in \Theta}{ \vdashc[\Theta] \get r : G}
			\and 
			\inferrule* [right=case]{\Gamma  \vdashv[\Theta] Z: \BB \\ \Delta  \vdashc[\Theta]  M : A \\ \Delta \vdashc[\Theta]  N : A}{\Gamma \boxplus \Delta \vdashc[\Theta]  \ifthen{Z}{M}{N} : A}	
			\and 
			\inferrule* [right=weak]{\Gamma  \vdashv[\Theta] V :A \\ x \text{ not free in } \Gamma}{\Gamma, x : P \vdashv[\Theta] V :A}
			\and 
			\inferrule* [right=weak]{\Gamma  \vdashc[\Theta] M :A \\ x \text{ not free in } \Gamma}{\Gamma, x : P  \vdashc[\Theta] M :A}
		\end{mathpar}
	\end{center}
	\end{tcolorbox}
	\vspace*{-2mm}
	\caption{$\lamBLL$ Typing Rules}
	\label{fig:typingRules}
\end{figure}
The typing rules for $\lamBLL$ are given in Figure \ref{fig:typingRules}. For a fixed reference context $\Theta  = r_1 : G_1, \dots, r_n : G_n$, we have two kinds of typing judgements, namely
$
\Gamma  \vdashv[\Theta] V : A \quad \text{and} \quad\Gamma \vdashc[\Theta]M:A
$
(for values and computations, respectively), whereas $\Gamma = x_1 : P_1, \dots, x_n : P_n$ is a context assigning positive types to term variables.
The operations of polynomial addition and multiplication induce a binary partial operation $\boxplus$ and a total unary operation $\ast$
on positive types defined by induction below:
\[\begin{aligned}
	G \boxplus G &:= G &&&	p\ast G &:= G\\
	(P \otimes Q) \boxplus(R \otimes S) &:= (P \boxplus R) \otimes (Q \boxplus S) &&&
	p \ast (P \otimes Q)  &:= (p \ast P) \otimes (p \ast Q) \\
	(\oc_p A) \boxplus (\oc_q A) &:= \oc_{p+q}  A &&&
	p \ast (\oc_q  A) &:= \oc_{p\times q}  A. 
\end{aligned}\]
Note that for ground types, the identities $G \boxplus G = G$ and $p \ast G = G$ justify duplicating ground values, 
while grades track the use of higher-order functions, following~\cite{dal2012linear}. 
The partial operation $\boxplus$ on positive types can itself be extended to a total operation on variable contexts:
\[\begin{aligned}
	\varnothing \boxplus \varnothing&:= \varnothing\\
	(x: P, \Gamma) \boxplus \Delta & := \begin{cases}
		x : P, \Gamma \boxplus \Delta & \text{ if } x \text{ does not occur in } \Delta\\
		x : P \boxplus Q, \Gamma \boxplus \Sigma & \text{ if } \Delta = x: Q, \Sigma
	\end{cases}
\end{aligned}\]
Similarly, the operation $p \ast (-)$ on positive types can be extended to a total operation on variables contexts
in a natural way. 

The substitution of all occurrences of $i$ with $p \in \N_{\geq 1}[i]$ within a typing context $\Gamma$ can be easily defined: simply,
the substitution is applied uniformly to all types occurring in $\Gamma$.
The following lemma establishes that typability is preserved under substitution of the security parameter by a polynomial. 
\begin{lemma}
		Let $s\in\{v,c\}$. Given a type derivation 
		$\pi \triangleright \Gamma \vdash^\Theta_s M:A$ and a polynomial $q\in \N_{\geq 1}[i]$, there exists a type 
		derivation $\pi q$ such that  $\pi q \triangleright \Gamma\subst{q} \vdash^{\Theta\subst{q}}_s M\subst{q}:A\subst{q}$.
\end{lemma}

For technical convenience, in the rest of this paper we will assume that any conditional $\ifthen{Z}{M}{N}$ is such that both branches $M$ and $N$ are \emph{closed} terms, i.e., $\mathit{FV}(M)=\mathit{FV}(N)=\emptyset$. This restriction entails no loss of generality since any conditional with open branches can be equivalently represented using $\lambda$-abstractions, preserving typing and semantics.  
\subparagraph*{Modeling Randomization and References}
The system $\lamBLL$ incorporates probabilistic effects and references to which one has to give semantics. We do so monads \cite{Moggi1991}, although we make only limited use of the underlying theory and assume just some familiarity with it.

Recall that for a set $X$, a \emph{(finite) probability distribution} is a function $\mu : X \to [0,1]$ with finite support, \ie the set $\supp{\mu} := \{ x \in X \mid \mu(x) > 0\}$
is finite, and such that $\sum_{x\in X} \mu(x) =1$. We denote by $\delta_x : X \to [0,1]$ the \emph{Dirac distribution} mapping an element $y$ in $X$ to $1$ if $y=x$ and to $0$ otherwise.
Any probability distribution $\mu$ is then equal to 
\(\sum_{x\in\supp{\mu}} \mu(x)\cdot\delta_{x}\), where we use the fact that a convex combination of
distributions is itself a distribution. We denote by $\D(X)$ the set of all probability distributions over $X$. This induces a monad $(\D, \eta_{\D}, \bind_{\D})$ on the category $\Set$ of sets and functions (we will omit the subscripts if there is no ambiguity). 

A \emph{store}, instead, is a map from memory reference variables to values which respects types. More precisely, for a fixed \emph{closed} reference context $\Theta = r_1 : G_1, \dots, r_m : G_m$ (meaning that the security parameter variable $i$ does \emph{not} occur in the types $G_1, \dots G_m$), we denote by $\Store_\Theta$ the set of functions $e : \{r_1, \dots, r_m \} \to \Val$ such that $e(r_j) \in \{ V \in \Val \mid \, \cdot \vdashv[\Theta] V : G_j\}$ for all $1 \leq j \leq m$. We associate to every closed $\Theta$ a corresponding monad $(\T_\Theta, \eta_\Theta, \bind_\Theta)$ on $\Set$ given by the tensor product \cite{hyland_combining_2006} of the distribution monad with the state monad $\T_{\Theta}: = (\D (- \times \Store_{\Theta}))^{\Store_{\Theta}}$, similarly to \cite{aguirre_higher-order_2021}. The unit of $\T_\Theta$ maps $X$ to $\D (X \times \Store_\Theta)^{\Store_\Theta}$ mapping $x \in X$ and $e \in \Store_\Theta$ to the Dirac distribution $\delta_{(x,e)}$.  The bind operator $
\bind_\Theta : \T_\Theta X  \times \Set(X, \T_\Theta Y) \to \T_\Theta Y
$ takes $\varphi \in \T_\Theta X $ and $f : X \to \T_\Theta Y$ to the map $\dlambda e. (\varphi(e) \bind_\D \eval \circ(f \times \id_{\Store_\Theta})) $ where $\dlambda$ and $\eval$ are respectively the Currying operator and the evaluation map induced by the Cartesian closed structure of $\Set$.  Explicitly, $\varphi \bind_\Theta f$ maps a store $e \in \Store_\Theta$ to the following distribution in $\D (Y \times \Store_\Theta)$:
\[
(y,e_1) \mapsto \sum_{(x, e_2)} \varphi(e)(x, e_2) \cdot f(x)(e_2)(y, e_1).
\]

\subparagraph*{From Sets to Indexed Families}
Working with general term sequents where the security parameter $i$ \emph{may occur freely}, requires generalizing the discussion above from sets to \emph{families} of sets. Let $\ISet$ be the category whose 
objects are families $X = \{X_n\}_{n\geq 1}$ of sets indexed by $\N_{\geq 1}$ and a morphism from $X = \{X_n\}_{n\geq 1}$ to $Y = \{Y_n\}_{n\geq 1}$ is a family of functions $\{f_n : X_n \to Y_n\}_{n\geq 1}$. Probabilistic effects are generated via the function symbols in $\Fun$. For each $f \in \Fun$ with $\typeof{f} = G_1 \times \dots \times G_m \to G$, we assume that there is a family $\den{f} = \{ \den{f}_n \}_{n \geq 1}$ of 
functions $\den{f}_n : \den{G_1}_n \times \dots \times  \den{G_m}_n \to \D(\den{G}_n)$
indexed over the security parameter $n \geq 1$, where $\den{\Str[p]}_n:= \{0, 1\}^{p(n)}$,  $\den{\BB}_n:= \{\true, \false\}$ and $\den{\UU}_n :=\{\star\}$.
Moreover, we assume that every $f\in\Fun$ can be evaluated in probabilistic polynomial time, namely that
there exists a PTM $\algof(f)$ such that for every $n\geq 1$, if $\algof(f)$ is fed with input $1^n$ and a tuple $t \in 
\den{G_1}_n \times \dots \times  \den{G_m}_n$, it  returns $x \in \den{G}_n$ with probability $\den{f}_{n}(t)(x)$.
Now, for a general reference context $\Theta$ (whose types may contain $i$), we define a monad on $\ISet$ mapping an indexed family $X = \{X_n\}_{n\geq 1}$ to the family
$
\{ \T_{\Theta n}(X_n)\}_{n\geq 1} =
\{
(\D (X_n \times \Store_{\Theta n}))^{\Store_{\Theta n}}
\}_{n\geq 1} ,
$
which we will use for the operational semantics of $\lamBLL$.
\subparagraph*{Operational Semantics}
We are now ready to present the operational semantics of $\lamBLL$, in the form of a small-step evaluation relation on computations. This evaluation system will be the primary object of study throughout the remainder of the paper, and it is with respect to it that we establish our soundness and completeness results.

For a variable context $\Gamma$, reference context $\Theta$ and type $A$, we define indexed families $\Comp{}{\Gamma}{\Theta}{A}= \{\Comp{n}{\Gamma}{\Theta}{A} \}_{n \geq 1}$ and $\Values{}{\Gamma}{\Theta}{A}= \{\Values{n}{\Gamma}{\Theta}{A} \}_{n \geq 1}$ of typable terms and values respectively as 
\[
	\Comp{n}{\Gamma}{\Theta}{A}:= \{ M \in \Lambda \mid  \Gamma\subst{n}  \vdashc[\Theta\subst{n}] M : A\subst{n} \}; \qquad
	\Values{n}{\Gamma}{\Theta}{A} := \{ V \in \Val \mid  \Gamma\subst{n}\vdashv[\Theta\subst{n}]  V : A\subst{n}\}.
\]
If the variable context $\Gamma$ is empty, we write $\CompC{n}{\Theta}{A}$ and $\ValuesC{n}{\Theta}{A}$ for $\Comp{n}{\varnothing}{\Theta}{A}$ and $\Values{n}{\varnothing}{\Theta}{A}$ respectively. 
For a fixed reference context $\Theta$ and type $A$, the small step operational semantics, see Figure~\ref{fig:smallstep}, is an indexed family of maps 
\[
\{\smstep[\Theta][A][n] :\CompC{n}{\Theta}{A}\times \Store_{\Theta \subst{n}} \longrightarrow \D_{\bot}(\CompC{n}{\Theta}{A} \times \Store_{\Theta\subst{n}})\}_{n\geq 1}.
\]
where $\D_{\bot}$ stands for the lifting of the monad $\D$. In  the case $\setr{r}{V}$, $e[V/r]$ denotes the store mapping a reference 
$s$ to $V$ if $s=$ and to $e(s)$ otherwise.
\begin{figure}[!h]
	\centering
	\begin{tcolorbox}[boxrule=0.6pt, colback=white]
		\scriptsize
		\vspace{-5pt}
		\begin{align*}
		\smstep[\Theta][A][n](\return V,e) &:= \bot \\
		\smstep[\Theta][A][n](MZ,e) &:=
		\begin{cases}
			\delta(N[Z/x],e)
			& \text{if } M=\return \lambda x.N, \\[2mm]
			\begin{aligned}
				&\smstep[\Theta][P \multimap A][n](M,e) \\
				&\qquad \bind_{\D_{\bot}}
				((N,e') \mapsto \delta(NZ,e'))
			\end{aligned}
			& \text{otherwise.}
			\end{cases}\\
			\smstep[\Theta][A][n](\letin{x=N}{M},e) &:=
			\begin{cases}
				\delta(M[V/x],e)
				& \text{if } N = \return V, \\[2mm]
				\begin{aligned}
					&\smstep[\Theta][P][n](N,e) 
					\bind_{\D_{\bot}} \\
					&\quad((N',e') \mapsto
					\delta(\letin{x=N'}{M},e'))
				\end{aligned}
				& \text{otherwise.}
			\end{cases}\\
			\smstep[\Theta][A][n](\letin{ \lis{x,y} = \lis{V,W}}{M}, e) &:= \delta(M[V/x,W/y], e)\\
			\smstep[\Theta][A][n](\der(\oc M), e)) &:= \delta(M,e)\\
			\smstep[\Theta][A][n]( \ifthen{\true}{M}{N} , e)&:= \delta(M,e)\\
			\smstep[\Theta][A][n]( \ifthen{\false}{M}{N} , e)&:= \delta(N,e)\\
			\smstep[\Theta][A][n](\setr{r}{V} , e)&:= \delta(\return \star,e[V/r])\\
			\smstep[\Theta][A][n](\get{r} , e)&:= \delta(\return e(r),e)\\
			\smstep[\Theta][A][n](f_p(Z_1, \dots, Z_m) , e)&:= (\return V,e')\mapsto (\delta(e)(e')\den{f}_{p(n)}(Z_1, \dots, Z_m)(V), e')\\
			\smstep[\Theta][A][n](\loopp{\lambda x. M}{1}{N} , e)&:= \delta(\letin{x=N}{M},e)\\
			\smstep[\Theta][A][n](\loopp{\lambda x. M}{k+1}{N} , e)&:= \delta(\letin{x=(\loopp{(\lambda x.M)}{k}{N})}{M},e)
		\end{align*}
	\end{tcolorbox}
	\caption{Small Step Semantics of $\lamBLL$}
	\label{fig:smallstep}
\end{figure}
Most of the rules are standard and do not involve effects directly, which come into play only when reducing functions in $\Fun$ and the constructs $\mathtt{get}$ and $\mathtt{set}$. 

Observe that $\mathbf{sm}$, as just defined, models only a \emph{single} reduction step. To reason about \emph{multi-step} evaluation, we introduce an iterated version of $\mathbf{sm}$, indexed by an additional parameter $k\in\N$. Formally, we define the \emph{iterated small step reduction} for every $k \in \N$, as an indexed family of partial functions 
\[
\{\smstep[\Theta][A][n,k] :\CompC{n}{\Theta}{A}\times \Store_{\Theta \subst{n}} \longrightarrow \D_{\bot}(\ValuesC{n}{\Theta}{A} \times \Store_{\Theta\subst{n}})\}_{n\geq 1}.
\]
whereas $\smstep[\Theta][A][n,0](M,e)$ is $\delta(V,e)$ if $ M =\return V$ and $\bot$ otherwise.
The following rule also holds:
\begin{mathpar}\mprset {sep=2em}
	\inferrule* []{ \smstep[\Theta][A][n](M,e) = \sum_k a_k \delta(M_k, e_k) \\ \smstep[\Theta][A][n,l_k](M_k,e_k) = \D_k
	}{\smstep[\Theta][A][n,1+ \max_k l_k](M,e)=\sum_k a_k \D_k }	
\end{mathpar}
In other words, the fact that $\smstep[\Theta][A][n,k](M,e)$ is a distribution witnesses that the evaluation of $(M,e)$ can be completed in at most $k$ steps, independently of the outcomes of probabilistic choices.

Given a pair $(M,e)\in\CompC{n}{\Theta}{A}\times \Store_{\Theta \subst{n}}$, what is the distribution over states to which this pair evaluates?
If there is $m\in\N$ such that $\smstep[\Theta][A][n,m](M,e)= \Dist$, then we
write $\itsmstep[\Theta][A][n](M,e)$ for such a distribution $\Dist$:
$
\itsmstep[\Theta][A][n](M,e)=\Dist\quad\Longleftrightarrow\quad\exists m.\smstep[\Theta][A][n,m](M,e)= \Dist.
$

\subparagraph*{Subject Reduction and Progress}

A form of \emph{type soundness} can be proved using standard techniques. The only difference with respect to the
classic forms of subject reduction and progress is that, since the operational semantics is monadic, the classical 
machinery (in particular, substitution lemmas) needs to be adapted to the presence of state and distributions. 
More details are in Appendix~\ref{app:progress}.

\section{On $\lamBLL$ as a Language For Cryptographic Proofs}\label{sec:cryptography}
After having introduced the essential features of $\lamBLL$, it is worth taking a moment to explain in what sense this formalism can serve as a framework for supporting cryptographic proofs. We do so by outlining the main ideas behind a proof presented in \cite{aplas}.

We briefly recall the notion of security we are interested in, following the presentation in Katz and Lindell \cite{katz2020introduction}, to which we refer for more details. The security notion of interest is indistinguishability under chosen-plaintext attacks (CPA security), formalized through the $\mathsf{PrivK}_{A,\Pi}$ experiment, in which a probabilistic polynomial-time adversary $A$ interacts with an encryption oracle for the encryption scheme $\Pi$ and attempts to distinguish encryptions of two chosen messages. One way of achieving CPA security consists in building $\Pi$ through pseudo-random functions (PRFs), namely efficiently computable functions that are computationally indistinguishable from truly random functions. Given such a function $F$, we write $\Pi_F$ for the induced encryption scheme. The proof that $\Pi_F$ is CPA-secure if $F$ is a PRF goes by reduction, and consists in building, out of an hypothetical adversary $A$ against $\Pi$, a successful distinguisher for $F$.

The proof can be formalized entirely within $\lamBLL$ by representing each component of the experiment $\mathsf{PrivK}_{A,\Pi}$ as a typed term.  The pseudorandom function $F$, key-generation algorithm $\Gen$ and encryption algorithm $\Enc$, are modelled by values of the following types
\[
\vdashv[\emptycontextCPA] F:\typstr{i}\tensor\typstr{i}\lollipop\typstr{i} 
\qquad
\vdashv[\emptycontextCPA] \Gen:\typunit \lollipop \typkeyprivk
\qquad
\vdashv[\emptycontextCPA] \Enc:(\typkeyprivk \tensor \typmsgprivk)\lollipop\typstr{2i},
\]
where $\Enc$ makes use of $F$ internally. These are not meant to be variables, but hypothetical $\lamBLL$ terms. The experiment $\mathsf{PrivK}_{A,\Pi}$ can itself be implemented as a $\lamBLL$ term which first draws a key at random, then asking the adversary $A$ to produce two messages $m_0$ and $m_1$ based on an encryption oracle, and then after encrypting $m_b$ (where $b$ is random) into $c$, asks $A$ to guess $b$ from $c$, still having access to the encryption oracle. As a consequence, then, an hypothetical adversary can be seen as a pair of terms
\begin{align*}
\vdashv[\emptycontextCPA]A_1&:
!_{q_1}(\typmsgprivk \lollipop \typstr{2i})
\lollipop
\typmsgprivk \tensor \typmsgprivk\tensor\typstr{r},\\
\vdashv[\emptycontextCPA]A_2&:
!_{q_2}(\typmsgprivk \lollipop \typstr{2i})\tensor\typstr{r}
\tensor\typstr{2i}
\lollipop
\typbool.
\end{align*}
where $r$ denotes the size of the adversary's internal state, while $q_1$ and $q_2$ are polynomial bounds on the number of queries that the adversary may issue to the encryption oracle during the two phases of the experiment.

The actual proof consists in constructing a distinguisher $D_A$ for $F$ that uses $A_1$ and $A_2$ as subroutines, and which can be given the type
\[
\vdashv[\emptycontextCPA]D_A:
!_{q_1+q_2+1}(\typstr{i}\lollipop\typstr{2i})
\lollipop
\typbool,
\]
The fact that this reduction is correct, namely that $D_A$ is successful whenever $A$ is successful, can in turn be formalized as a game-based proof \cite{Shoup2004} using the equational theory introduced in \cite{aplas}, which is based on a form of approximate logical relations. Notice, however, that neither $A_1$ nor $A_2$ are concrete $\lamBLL$ terms: they merely stand for hypothetical adversaries whose existence is assumed for the sake of the reduction proof.

There is one missing piece in this argument. How can we be sure that the construction of $D_A$ from $A_1$ and $A_2$, which is nothing more than a well-typed $\lamBLL$ term with two unspecified subterms, is \emph{complexity preserving}? To establish it, one must show that every instance of $D_A$ obtained by taking $A_1$ and $A_2$ as concrete, efficient algorithms is itself efficient. This is precisely where the soundness and completeness results become useful: if $A_1$ and $A_2$ are efficient, we can assume (by \emph{completeness}) that terms equivalent to them can be written in $\lamBLL$, and plugged into $D_A$. We can thus conclude (by \emph{soundness}) observing that now the term $D_A$ is a concrete term of $\lamBLL$, and can thus be evaluated in polynomial time.
\section{Soundness and Completeness for Polynomial Time}\label{sect:soundcomplete}
This is certainly the most important section of this paper, namely the one in which we prove that $\lamBLL$ is sound and complete for complexity classes that faithfully capture the notion of efficiency used by cryptographers. In the previous section, we argued why reaching such results would be useful. 

We begin, in the first subsection, by addressing polytime soundness; this is achieved by studying the dynamics of the system and providing appropriate bounds on the computational complexity of simulating the evaluation process.

\subsection{Soundness}

In order to prove polynomial-time soundness, we show that there exists a polynomial bound on the length of reduction sequences for any term in $\lamBLL$, and that term size does not increase too much along those reductions. First, we define the size of terms and values in $\lamBLL$: each constructor contributes for $1$, and the total size is thus obtained by counting the number of constructors, keeping in mind that strings count for their length. This way we 
can easily define the natural numbers $\sizeM{M}$ and $\sizeV{V}$ whenever $M$ is a computation and 
$V$ is a value. We also give a formal definition of the length of the reduction of a term $M$ and store $e$:
\begin{definition}[Cost]
	\label{def:cost}
	For a term $M \in  \CompC{n}{\Theta}{A} $ and a store $e \in \Store_{\Theta n}$, we define the \emph{evaluation cost} $\cost_n^{\Theta;A}(M,e) \in \N\cup \{+\infty\}$ as
	\[\cost_n^{\Theta;A}(M,e)
	:=
	\inf \{ k\mid\smstep[\Theta][A][n,k](M,e) \neq \bot \}.
	\] 
\end{definition}
This captures the \emph{maximal} number of steps required to reach a value. Our goal is to show that, for every well-typed term $M$, this  is polynomially bounded in $n$, uniformly across all valid stores. 
In order to prove the latter, we proceed by associating to each type derivation $\pi$ a polynomial $\poly{\pi} \in \mathbb{N}_{\geq 1}[i]$ that serves as a bound on the computational complexity of evaluating the underlying term. This polynomial is defined so that it satisfies \emph{two} key properties: it upper-bounds the evaluation cost as from Definition~\ref{def:cost}, but it also dominates the size of the term. 
As a result, we conclude that all typable terms in $\lamBLL$ can be evaluated in polynomial time, and thus the language is sound for polynomial-time computation.

 The key towards defining $\poly{\cdot}$ in an effective way consists in properly giving a weight to term \emph{variables}. In particular, we assign a polynomial to each positive type in $\lamBLL$ by induction,  reflecting the structure of types, but \emph{not} counting higher-order subtypes: ground types such as $\UU$ and $\BB$ have a fixed size equal to 1, string types $\Str[p]$ contribute a size equal to the parameter $p$, tensor products sum up their size, while exponential types are assigned a minimal positive size. Formally:
\begin{align*}
			\sizeG{\UU}&:=1 & 	\sizeP{P \otimes Q}&:=\sizeP{P}+\sizeP{Q}+1\\
			\sizeG{\BB}&:=1 & \sizeZ{!_p A}&:= 1 \\
			\sizeG{\Str[p]}&:=p & &
\end{align*}
We are now ready to define the function $\poly{\cdot}$  by induction on the structure of 
the argument type derivation $\pi$ (some rules can be found in Figure~\ref{fig:poly-cost}). 
\begin{figure}[htbp]
	\begin{tcolorbox}[boxrule=0.6pt, colback=white]
	\scriptsize
		\begin{mathpar}\mprset{sep=1em,  vskip =0.01ex}
			\polyInferRule{\inferrule* [right=var]{~}{x:P \vdashv[\Theta] x : P}}
			:= \sizeP{P}
			
			\polyInferRule{\inferrule* [right=unit]{~}{\vdashv[\Theta] * : \mathbb{U}}}
			:= 1
			
			\polyInferRule{\inferrule* [right=string]{
					s \in \{0,1\}^c \\p : i \mapsto c \text{ is a const. polynomial}}{\vdashv[\Theta] s : \mathbb{S}[p]}}
			:= p
			
			\polyInferRule{\inferrule* [right=tensor]{
					\pi_1 \triangleright \Gamma \vdashv[\Theta] Z_1 : P \\
					\pi_2 \triangleright \Delta \vdashv[\Theta] Z_2 : Q
				}{\Gamma \boxplus \Delta \vdashv[\Theta] \lis{Z_1, Z_2} : P \otimes Q}}
			:=1+\poly{\pi_1} + \poly{\pi_2}
			
			\polyInferRule{\inferrule* [right=bang]{\pi\triangleright \Gamma \vdashc[\Theta] M : A}{p \ast \Gamma \vdashv[\Theta] \oc M : \oc_p A}}
			:= p \times \poly{\pi}+p
			
			\polyInferRule{\inferrule* [right=lam]{\pi\triangleright \Gamma, x: P \vdashc[\Theta] M : A}{\Gamma \vdashv[\Theta] \lambda x. M : P \multimap A}}
			:= 1 + \poly{\pi}
						
			\polyInferRule{\inferrule* [right=let]{\pi_1 \triangleright \Gamma  \vdashc[\Theta] N: P\quad \pi_2 \triangleright x : P, \Delta  \vdashc[\Theta] M : A }{ \Gamma\boxplus \Delta  \vdashc[\Theta] \letin{ x = N}{M} : A}}
			:= 1+\poly{\pi_1}+\poly{\pi_2}
			
			\polyInferRule{\inferrule* [right=let\_pair]{\pi_1 \triangleright\Gamma \vdashv[\Theta] Z :  P \otimes Q\\x : P, y : Q, \pi_2 \triangleright\Delta \vdashc[\Theta]  M : A}{\Gamma \boxplus \Delta \vdashc[\Theta]  \letin{ \lis{x,y} = Z}{M} : A }}
			:= 1+\poly{\pi_1}+\poly{\pi_2}
			
			\polyInferRule{\inferrule* [right=case]{
					\pi_g \triangleright\Gamma  \vdashv[\Theta] Z: \BB \\ \pi_1 \triangleright \Delta  \vdashc[\Theta]  M : A \\ \pi_2 \triangleright\Delta \vdashc[\Theta]  N : A
				}{\Gamma \boxplus \Delta \vdashc[\Theta]  \ifthen{Z}{M}{N} : A}}
			:= 
			     1 + \poly{\pi_g} +\poly{\pi_1}+\poly{\pi_2}
			
			\polyInferRule{\inferrule* [right=loop]{
					\pi_1\triangleright \Gamma \vdashv[\Theta] V: P\multimap P \\ 
					\pi_2\triangleright \Delta  \vdashc[\Theta]  M : P
				}{(p \ast \Gamma)\boxplus \Delta  \vdashc[\Theta]  \loopp{V}{p}{M} : P }}
			:=  1+p+p\times \poly{\pi_1}+\poly{\pi_2}
					
		\end{mathpar}	
		\end{tcolorbox}
	\caption{The Function $\poly{\cdot}$}
	\label{fig:poly-cost}
\end{figure}
For basic ground values (booleans, unit, strings), the polynomial is given by their size, which is pretty natural. For composite values, like pairs, abstractions, and applications, the polynomial cost is defined as the sum of the polynomials of their sub-derivations, adding a constant overhead. 
The \textsc{BANG} rule introduces multiplicative cost: promoting a computation $M$ to a value of type 
$\oc_p A$ multiplies its cost by the resource polynomial $p$, reflecting potential duplication. As mentioned above, what might seem a bit strange at first glance is the way variables are weighted. Although the syntactic size of a variable is small (after all, $\sizeM{x}=1$ for every variable $x$), the weight of an instance of \textsc{VAR} is given by the size $\sizeP{P}$ of the underlying positive type $P$, which does not take into account the size of the part of $P$ that lies within the scope of a modality $!_p$. This is due to the fact that while the type system allows the indiscriminate duplication of ground-type values, the type $!_p A$ implicitly accounts for the number of copies of such value potentially needed. The two must then be treated differently.

As we mentioned above, the polynomial cost $\poly{\pi}$ associated to a typing derivation $\pi$ is designed to upper-bound the evaluation cost of the corresponding term and to dominate its syntactic size. This latter property is formalized by the following lemma.
\begin{proposition}[Size Soundness]
	\label{lemma:monotonicity_poly_size}
	For every type derivation $\pi\;\triangleright  \vdashc[\Theta] M:A$, we have that 
	$\poly{\pi}(n) \geq \sizeM{M\subst{n}}$ 
	for every $n\geq 1$. Similarly for values:  for every type derivation 
	$\pi\;\triangleright  \vdashc[\Theta] V:A$ we have that $\forall n\geq 1.\mbox{ } \poly{\pi}(n) \geq 
	\sizeV{V\subst{n}}$ for every $n\geq 1$.
\end{proposition}
This result can be proved by induction on $\pi$. This is of course not enough. We have to prove that
$\poly{\pi}$ constitutes a bound on the number of reduction steps necessary to reduce the
underlying term to normal form.
\begin{restatable}[Step Soundness]{theorem}{PolySoundness}\label{th:polysound}
	For every type derivation $\pi\;\triangleright  \vdashc[\Theta] M:A$, it holds that
	$\cost_n^{\Theta;A}(Mn,e) \leq \poly{\pi}(n)$
	for every $n\geq 1, e\in \Store_{\Theta n}$.
\end{restatable}
As detailed in the Appendix, the proof of this Theorem requires a nontrivial machinery, and in particular
a nonstandard Substitution Lemma. The substitution process, a cornerstone of any subject reduction 
proof, takes a quantitative form here, which however cannot be the natural one (which simply does not go
through in the induction). Instead, it is necessary to prove a stronger result, namely that it is sufficient to consider only the higher-order 
part of the substituted value, delegating to the second term the task of balancing the rest through 
its free variable occurrences.

Putting Lemma \ref{lemma:monotonicity_poly_size} and Theorem \ref{th:polysound} together, we get:
\begin{corollary}[Polytime Soundness]\label{cor:polytimesoundness}
	For every type derivation $\pi\;\triangleright  \vdashc[\Theta] M:A$, there is a polynomial $p$
	such that the problem of evaluating configurations in the form $(Mn,e)$ takes time polynomial in $n$.
\end{corollary}
\begin{proof}
It is sufficient to consider the algorithm that, given $n$ and $(Mn,e)$ as inputs, proceeds by simulating step-by-step the evaluation of $(Mn,e)$ through the operational semantics $\smstep[\Theta][A][n]$. Whenever this process requires the evaluation of a symbol $f$ in $\Fun$, one will obviously resort to the underlying algorithm $\algof(f)$, which by hypothesis runs in polynomial time. The fact that the number of macro-steps is bounded by a polynomial is a consequence of Theorem \ref{th:polysound}, while the fact that each step takes polynomial time is a consequence of the fact that the terms involved cannot be too large, thanks to Lemma \ref{lemma:monotonicity_poly_size}.
\end{proof}

\subsection{Completeness}\label{subsect:completeness}

The goal of this section is to establish that every function computable in polynomial time by a Probabilistic Turing Machine can be represented by a $\lamBLL$ term computing the same function. Remarkably, this can be achieved using only a finite and minimal set of basic function symbols, consisting of cyclic shifts and bit-reading and bit-writing operations on strings. This shows that the bulk of the computation is carried out by $\lamBLL$ itself, rather than being delegated to the function symbols.

We first of all need to define what we mean by a probabilistic polynomial time function, and in doing so we follow the discussion in Section \ref{sect:birdseye}. In particular, the mathematical objects that will become elements of our complexity class will not be mere functions over natural numbers, but families of random functions, i.e., families of functions each of which, given a string as input, returns a \emph{distribution} of strings. Formally, given a \emph{first-order interface} consisting of a pair of polynomials $I=(p,q)$, an \emph{I-function} $f$ is a family of functions:
\[
f=\{f_n: \{0,1\}^{p(n)} \rightarrow \D(\{0,1\}^{q(n)}) \}_{n\in\mathbb{N}}.
\]
Noticeably, $\lamBLL$ terms of first-order type naturally computes certain $I$-functions.
Given $\vdashv[\varnothing] M: \Str[p] \lollipop \Str[q]$, we say that $M$  \emph{computes} a $(p,q)$-function if  
$\itsmstep[\varnothing][{\Str[q]}][n]((Mn)s,\emptyset)=(f_n(s)(t),\emptyset)$ for every $n\geq 1, s\in\{0,1\}^{p(n)}$.

In the following, we are going define the class of functions $\mathsf{PPT}(I)$ whose elements are the $I$-functions that we consider efficiently computable: however, to arrive at a formal definition, despite the intuitive simplicity, it is necessary to recall the reference model of computation, namely that of probabilistic Turing machines.

A Probabilistic Turing Machine (PTM, \cite{Gill1977}) is a generalization of the classical deterministic Turing Machine having access to a source of randomness, and in particular to a potentially infinite amount of independent uniformly distributed bits, based on which the machine can make decisions. We consider a particular model of PTM with three tapes: input tape, work tape and output tape. All tapes operate over a common tape alphabet $\tapeAlph=\{0,1,\blank\}$, where $\blank$ denotes the blank symbol. We also define the input/output alphabet as the set $\Sigma=\{0,1\}$.

Formally, we define a PTM as $\mathcal{M}=(Q,\delta_0,\delta_1,q_0, q_{f})$, where $Q$ denotes a finite set of states, $q_0 \in Q$ is the initial state, and $q_f \in Q$ is the final (or accepting) state. The behavior of the machine is determined by two transition functions, $\delta_0$ and $\delta_1$, which specify the next move depending on the outcome of a fair coin flip. At each step of the computation, the machine is in a state $q\in Q$ and reads a pair of symbols $(\alpha_{\mathrm{in}},\alpha_{\mathrm{work}})\in\tapeAlph\times\tapeAlph$ corresponding to the current symbol under the head of the input and work tape, respectively.
The machine flips a fair coin and depending on the output $b\in\{0,1\}$ of it, applies the corresponding transition $\delta_b$. A configuration for a PTM represents its computational state at a specific moment during its execution. Formally, a configuration is a tuple
\[
C=(q, (l_{\mathrm{i}},symb_{\mathrm{i}},r_{\mathrm{i}}), (l_{\mathrm{w}},symb_{\mathrm{w}},r_{\mathrm{w}}), (l_{\mathrm{o}},symb_{\mathrm{o}},r_{\mathrm{o}}))
\]
where $q$ is the current state and every tape is represented by a triple $(l,symb,r)$ composed of the following : the segment of the tape to the left of the head, the symbol currently under the head, the segment of the tape to the right of the head. The way the underlying configuration evolves can be defined formally as follows.
\begin{definition}[Evaluation Relation of a PTM]
	\label{def:stepRelationPTM}
	Let $\ConfSet$ denote the set of all possible configurations of a PTM. 
	We define a family of relations 
	$\Rightarrow \;\subseteq\; \ConfSet \times \mathbb{N} \times \D(\ConfSet)$
	as follows.
	\begin{mathpar}
		\scriptsize
		\mprset{sep=2em}
		\inferrule* [right=Final]
		{n \in \mathbb{N} \\ \text{$C$ is final}}
		{C \xRightarrow{n} \{C^1\}}
		\qquad\qquad
		\inferrule* [right=Step]
		{\delta_0(C) = D \\ \delta_1(C) = E \\ D \xRightarrow{n} \mathcal{D} \\ E \xRightarrow{n} \mathcal{E}}
		{C \xRightarrow{n+1} \tfrac{1}{2}\,\mathcal{D} + \tfrac{1}{2}\,\mathcal{E}}
	\end{mathpar}
\end{definition}
Moreover, we say that a PTM works in time $r$ if for every input $s\in\{0,1\}^n$, the initial configuration $C_s$ for $s$ is such that $C_s \xRightarrow{r(n)} \mathcal{E}$ for some $\mathcal{E}\in\D(\ConfSet)$. Given a first-order interface $I=(p,q)$, we say that a PTM computes an $I$-function $f$ in time $g:\N\rightarrow\N$ if for every input $s\in\{0,1\}^{p(n)}$, if the initial configuration $C_{(1^n,s)}$ is such that there exists $m \in \N$ with $C_{(1^n,s)}\xRightarrow{m} \mathcal{E}$, where $\forall t\in \{0,1\}^{q(n)}.f_n(s)(t)=\sum_{C \text{ final configuration for t}} \mathcal{E}(C)$, and $m\leq g(n)$.

We are finally able to define the complexity class that represents our target. 
The class $\mathsf{PPT}(I)$ includes those $I$-functions which can be computed by PTMs working in polynomial time.
Noticeably, this is the class that $\lamBLL$ captures.
 
\begin{restatable}[First-Order Completeness]{theorem}{Completeness} 
	\label{th:PPTCompleteness}
	For every first-order interface $I$, the class of $I$-functions which can be computed
	by $\lamBLL$ terms is precisely $\mathsf{PPT}(I)$.
\end{restatable}
\begin{proof}[Sketch]
		The fact that every first-order term in $\lamBLL$ computes a function in the underlying class $\mathsf{PPT}(I)$
		is an easy consequence of Corollary \ref{cor:polytimesoundness}.
		It thus remain to show that polytime PTMs can be embedded into $\lamBLL$.
		Here we only provide some hints for the proof, which can be found in Appendix \ref{app:completeness}. 
		The structure of the embedding is fairly standard, but it should be noted that:
	\begin{itemize}
		\item
		We need to handle the fact that the alphabet on which PTMs work has cardinality $3$, and therefore it is necessary to encode a symbol using multiple booleans, which is standard.
		\item
		Care must be taken due to the fact that the string type is sized, and therefore all strings must have a statically bounded length. However, handling this is not too difficult, because the PTM we want to encode works in polynomial time, and thus cannot use more than a polynomially bounded portion of their tapes.
		\item
		We also need to manage the fact that we want $\Fun$ to contain very few function symbols, namely cyclic shifts and functions for testing and setting the first bit in a string. This limits us, but given the locality of PTMs, this is not an issue either.
	\end{itemize}
\end{proof}

\section{Beyond First-Order}\label{sect:secondorder}

In the previous section, we showed that $\lamBLL$ 
is sound and complete with respect to a notion of polynomial-time 
computable function that faithfully captures cryptographic practice. As already 
mentioned in Section \ref{sect:birdseye}, however, not all cryptographic constructions
and adversaries have a first-order nature, as witnessed by the example in Section \ref{sec:cryptography}.

Let us suppose to modify the notion of 
polynomial-time computable function as follows. A 
\emph{second-order interface} is a pair of first-order interfaces $(I,J)$.
Given such a pair $P=((p,q),(r,s))$ a \emph{$P$-function} $f$ consists
of a family of functions of the following form:
\[
f=\{f_n: (\{0,1\}^{p(n)} \rightarrow \D(\{0,1\}^{q(n)}))\times\{0,1\}^{r(n)} \rightarrow \D(\{0,1\}^{s(n)}) \}_{n\in\mathbb{N}} .
\]
A function $f$ of this kind is naturally computed by a PTM of the sort introduced in Section 
\ref{sect:soundcomplete}, but with the additional capability of interacting with an oracle. The oracle behaves 
according to the interface $I=(p,q)$ but is otherwise an arbitrary function. 
The runtime of such a PTM can be defined in the usual way,
with oracle calls having unit cost. All in all, this allows us to introduce the
class $\mathsf{PPT}(P)$.

The definition we have just given intuitively corresponds to 
$\lamBLL$ terms whose type is of the form
$!_{t}(\Str[p]\lollipop\Str[q])\otimes\Str[r]\lollipop\Str[s]$.
 How, however, can we prove that this is indeed the case? There is evidently a problem in the very 
definition of computability for $\lamBLL$ terms of type like the one above. 
This issue is rooted in the fact that, in our calculus, at least in the form in which it has been 
used so far, it is not possible to test how such terms behave when fed with (first-order) 
functions \emph{not themselves representable} in the language.

\newcommand{\AFun}{\mathcal{A}}
Accordingly, it is necessary (solely for the purpose of extending the soundness and completeness result 
in the direction we are considering) to extend $\lamBLL$ as follows:
\begin{varitemize}
	\item
	On the one hand, the language of \emph{computations} 
	is extended with symbols that are syntactically treated as function symbols but that
	semantically serve to model \emph{oracle} calls. This form a set $\AFun$.
	\item
	On the other hand, these new symbols are given semantics 
	in a novel way. In particular, unlike the symbols in $\Fun$, 
	it is not assumed that any symbol $f\in\AFun$ has a fixed semantics $\den{f}$ 
	and a fixed evaluation algorithm $\algof(f)$. Instead, the semantics $\den{f}$ is passed 
	as an additional parameter to the operators 
	$\mathbf{sm}$ and $\mathbf{ism}$, which in this way will be able to 
	understand how to evaluate $f$.
	\item
	Finally, the weight
	$\mathit{poly(\pi)}$
	 of a derivation 
	$\pi$ consisting of an instance of a functional symbol in $\AFun$ 
	is defined to be equal to $1$ plus the weight of the immediate subderivations.
\end{varitemize}
With these three new ingredients, it is finally possible to define when a term of the above kind computes a 
$P$-function. More importantly, it becomes possible to prove that the resulting class of functions coincides with that of 
$P$-functions computable in polynomial time:
\begin{theorem}[Second-Order Completeness]
	For every second-order interface $P$, the class of $P$-functions which can be computed
	by $\lamBLL$ terms is $\mathsf{PPT}(P)$.
\end{theorem}
\begin{proof}[Proof Sketch]
	The fact that every second-order $\lamBLL$ term computes a function in $\mathsf{PPT}(P)$ for an appropriate second-order interface $P$ can be proven by providing size and cost soundness results, following what we did in the first-order case. The crucial point is the lifting of Corollary \ref{cor:polytimesoundness}, where it must be noted that the simulation of the operational semantics can only be performed on oracle PTMs, given that this operational semantics depends on the interpretation of the symbols in $\AFun$. As for the encoding of PTMs, the procedure is similar to the first-order case. However, it is necessary to simulate the oracle call, and to do this, we proceed as follows:
	\begin{itemize}
		\item
		Of course, the PTM to be encoded has additional tapes available for the oracle call, but this does not pose any particular problem, because adding two more tapes to the pictures does not represent a substantial change.
		\item
		The simulation of the oracle call consists in passing the content of the corresponding input oracle tape to a first-order function, and then transferring the result back to the output oracle tape.
	\end{itemize}
	In other words, the proof follows very closely the structure of the one from Section \ref{sect:soundcomplete}.
\end{proof}
\section{Conclusion}\label{sect:conclusion}
In this work, we have shown that $\lamBLL$ \cite{aplas} is sound and complete for a notion of probabilistic polynomial time corresponding to the cryptographic practice, both at first order and at second order. To the best of the authors' knowledge, this is the first characterization of a notion of polynomial time specifically designed with cryptography in mind and in the spirit of implicit computational complexity.

\subsubsection*{Related Work}
A number of logics, programming languages, and formal systems have been designed to support cryptographic proofs and to capture probabilistic polynomial time. The first example in this direction is the logic, together with its associated simple programming language, introduced by Impagliazzo and Kapron~\cite{impagliazzo_kapron}. Along the same lines, one should certainly mention the works of Nowak and Zhang~\cite{nowak2010calculus}. In all these cases, however, the notion of probabilistic polynomial time employed is the classical first-order one, and there are no results concerning the correspondence with second-order polytime.
Broadening the scope, one should also cite the contributions on characterizations of probabilistic polynomial time via ICC techniques (see, e.g.,~\cite{DalLagoKahleOitavem2021}). Once again, the emphasis is on the characterization of classes such as $\mathsf{PP}$, $\mathsf{BPP}$ and $\mathsf{ZPP}$, which differ from the classes considered in this work.

\bibliographystyle{splncs04} 
\bibliography{biblio} 


\newpage
\appendix     
\section{Progress}
\label{app:progress} 
For a variable context $\Gamma$, a location context $\Theta$ and a type $A$, let $\CompO{\Gamma;}{\Theta}{A}$ and $\ValuesO{\Gamma;}{\Theta}{A}$ be the sets of derivable computation terms and values respectively which are open for the security parameter variable $i$:
\[
\Lambda_o(\Gamma; \Theta; A) := \{ M  \in \Lambda \mid  \Gamma   \vdashc[\Theta] M : A \} \qquad \Val_o(\Gamma; \Theta; A) := \{ V \in \Val \mid  \Gamma \vdashv[\Theta]  V : A\}.
\]

In particular, it requires that for every positive natural number $n$, the term obtained by substituting $n$ for the security parameter behaves correctly according to the semantics. This is formalized as follows

\begin{theorem}[Progress]\label{th:progress}
    For every $M\in \Lambda_o(\varnothing; \Theta; A)$, either
    \begin{enumerate}
        \item $M=\return V$ for some value $V\in \Val_o(\varnothing; \Theta; A)$ or
        \item $\forall n \geq 1$, $\forall e \in \Store_{\Theta n}$, $\exists \Dist\in \D(\CompC{n}{\Theta}{A} \times \Store_{\Theta\subst{n}})$ such that $\smstep[\Theta][A][n](Mn,e)=\Dist$.
    \end{enumerate}
\end{theorem} 

\begin{proof}
	By induction on the structure of a type derivation for $M$, by cases.
	\begin{itemize}
		\item \textit{Rule ETA}: We have that $M=\return{V}$, where $V\in \Val_o(\varnothing; \Theta; A)$ and we can conclude. 
    By induction on the derivation of the typing judgment $\vdashc[\Theta n] Mn: An$
    \\item.
        \item \textit{Rule FUN}: 
        \begin{align*}
            \inferrule* { \mathrm{typeof}(f) = G_1 \times \dots \times G_m \to G \\ ( \vdashv[\Theta n] Z_k : G_k\subst{pn})_{1\leq k \leq m} \\ p \in \N_{\geq 1}[i]}{  \vdashc[\Theta n] f_p(Z_1, \dots, Z_m):G\subst{p}n}
        \end{align*}
        If $\den{f}_{p(n)} : \den{G_1}_{p(n)} \times \dots \times  \den{G_m}_{p(n)} \to D(\den{G}_{p(n)})$, we have that \\ $\den{f}_{p(n)}(Z_1,\ldots,Z_m)=\sum_{j=1}^{m} a_j \mbox{ } \delta_{M_j}$. By operational semantics we have\\ $\forall e. \smstep[\Theta][A][n](f_p(Z_1, \dots, Z_m),e)=\sum_{j=1}^{m} a_j \mbox{ } \delta(M_j,e')$ and we can conclude.
        \\
        \item \textit{Rule LET\_PAIR}:
        \begin{align*}
            \inferrule* { \vdashv[\Theta n] Z :  Pn \otimes Qn\\x : Pn, y : Qn \vdashc[\Theta n]  M : An}
            {\vdashc[\Theta n]  \letin{ \lis{x,y} = Z}{M} : An }
        \end{align*}
        We must have $Z=\lis{Z_1,Z_2}$ such that $\vdashv[\Theta n] Z_1 : Pn$ and $\vdashv[\Theta n] Z_2 : Qn$ are derivable, then by small step semantics we have that  $\forall e. \smstep[\Theta][A][n](\letin{ \lis{x,y} = \lis{Z_1, Z_2}}{M},e)=\delta(N[Z_1/x, Z_2/y], e)$ and we can conclude.\\

        \item \textit{Rule DER}:
        \begin{align*}
            \inferrule* { \vdashv[\Theta n] Z :  !_1 An}
            {\vdashc[\Theta n]  \der(Z): An }
        \end{align*}
        We must have $Z=!M$ such that $\vdashc[\Theta n] N : An$ is derivable, then by small step semantics we have that  $\forall e. \smstep[\Theta][A][n](\der(!N),e)=\delta(N, e)$ and we can conclude.\\

        \item \textit{Rule LET}:
        \begin{align*}
            \inferrule* {\vdashc[\Theta n] N_1 :  Pn \\x : Pn \vdashc[\Theta n]  N_2: An}
            {\vdashc[\Theta n]  \letin{x = N_1}{N_2} : An }
        \end{align*}
        By induction hypothesis on $N_1$ we have that either $N_1=\return{V}$ or $\forall e \in \Store_{\Theta\subst{n}}. \exists D \in \Dist_{\bot}(\CompC{n}{\Theta}{A} \times \Store_{\Theta\subst{n}})$ such that $\smstep[\Theta][A][n](N_1)=D$. We proceed by analyzing these two cases:
        \begin{itemize}
            \item Sub-case $N_1=\return{V}$:\\
            We have that $M=\letin{x = N_1}{N_2}=\letin{x = \return{V}}{N_2}$. By hypothesis of LET-typing rule we have that $\vdashc[\Theta n] N_1 :  Pn$ is derivable, then in this case $\vdashc[\Theta n] \return{V} :  Pn$ is derivable.
            By small step semantics we have that  $\forall e. \smstep[\Theta][A][n](\letin{x = \return{V}}{N_2},e)=\delta(N_2[V/x], e)$ and we can conclude.
            \item Sub-case $D=\sum_{k} a_k \delta(M_k, e_k)$ such that $\smstep[\Theta][A][n](N_1)=D$:\\ 
            By small step semantics we have that $\forall e. \smstep[\Theta][A][n](\letin{x = N_1}{N_2},e)=\sum_{k} a_k \mbox{ }\delta(\letin{x=M_k}{N_2}, e_k)$ and we can conclude.
        \end{itemize}
        \vspace{3mm}
        \item \textit{Rule APP}:
        \begin{align*}
            \inferrule* {\vdashc[\Theta n] N:  Pn \multimap An \\  \vdashv[\Theta n]  Z: Pn}
            {\vdashc[\Theta n]  NZ : An }
        \end{align*}
        By induction hypothesis on $N$ we have that either $N=\return{V}$ or $\forall e \in \Store_{\Theta\subst{n}}. \exists D \in \Dist_{\bot}(\CompC{n}{\Theta}{A} \times \Store_{\Theta\subst{n}})$ such that $\smstep[\Theta][A][n](N)=D$. We proceed by analyzing these two cases:
        \begin{itemize}
            \item Sub-case $N_1=\return{V}$:\\
            We have that $M=NZ=(\return{V})Z$. By hypothesis of APP-typing rule we have that $\vdashc[\Theta n] N:  Pn \multimap An$ is derivable, then in this case $\vdashc[\Theta n] \return{V}:  Pn \multimap An$ is derivable, so $V$ is in the form $V=\lambda x.N'$ for some $N'$ such that $x:Pn\vdashc[\Theta n] N': An$.
            By small step semantics we have that  $\forall e. \smstep[\Theta][A][n]((\return{\lambda x. N'})Z,e)=\delta(N'[Z/x], e)$ and we can conclude.
            \item Sub-case $D=\sum_{k} a_k \delta(M_k, e_k)$ such that $\smstep[\Theta][A][n](N)=D$:\\ 
            By small step semantics we have that $\forall e. \smstep[\Theta][A][n](NZ,e)=\sum_{k} a_k \mbox{ }\delta(M_k Z, e_k)$ and we can conclude.
        \end{itemize}
        \vspace{3mm}
        \item \textit{Rule LOOP}:
        \begin{align*}
            \inferrule* {\vdashv[\Theta n] V:  Pn \multimap Pn \\  \vdashc[\Theta n]  N: Pn}
            {\vdashc[\Theta n]  \loopp{V}{p}{N} : Pn }
        \end{align*}
        By hypothesis of LOOP-typing rule we have that $\vdashv[\Theta n] V:  Pn \multimap Pn$ is derivable, s $V$ is in the form $V=\lambda x.N'$ for some $N'$ such that $x: Pn\vdashc[\Theta n]  N': Pn$ derivable. By small step semantics we have to analyze the following two sub-cases:
        \begin{itemize}
            \item Sub-case $p=1$:\\
            By small step semantics we have that $\forall e. \smstep[\Theta][A][n](\loopp{\lambda x.N'}{1}{N},e)=\delta(\letin{x=N'}{N},e)$ and we can conclude.
            \item Sub-case $p=1+q$ for some $q\in\N_{\geq 1}[i]$:\\
            By small step semantics we have $\forall e. \smstep[\Theta][A][n](\loopp{\lambda x.N'}{1+qn}{N},e)=\\\delta(\loopp{\lambda x.N'}{qn}{N},e)$ 
        \end{itemize}
        \vspace{3mm}
        \item \textit{Rule CASE}:\\
        \begin{align*}
            \inferrule* {\vdashv[\Theta n] Z:\mathbb{B}\\ 
            \vdashc[\Theta n]  N_1: An\\
            \vdashc[\Theta n]  N_2: An}
            {\vdashc[\Theta n]  \ifthen{Z}{N_1}{N_2} : An }
        \end{align*}
        By hypothesis of CASE-typing rule we have that $\vdashv[\Theta n] Z:\mathbb{B}$ is derivable, so we have to analyze the following two sub-cases:
        \begin{itemize}
            \item Sub-case $Z=\true$:\\
            By small step semantics we have $\forall e. \smstep[\Theta][A][n](\ifthen{\true}{N_1}{N_2},e)=\delta(N_1,e)$ and we can conclude. 
            \item Sub-case $Z=\false$:\\
            By small step semantics we have $\forall e. \smstep[\Theta][A][n](\ifthen{\false}{N_1}{N_2},e)=\delta(N_2,e)$ and we can conclude. 
        \end{itemize}
        \vspace{3mm}
        \item \textit{Rule SET}:
        \begin{align*}
            \inferrule* {\vdashv[\Theta n] Z:G \\ r:G \in\Theta n}
            {\vdashc[\Theta n]  \setr{r}{Z} : \mathbb{U} }
        \end{align*}
        By small step semantics we have $\forall e. \smstep[\Theta][A][n](\setr{r}{Z},e)=\delta(\return{\ast},e[Z/r])$ and we can conclude. \\
        \item \textit{Rule GET}:
        \begin{align*}
            \inferrule* {r:G \in\Theta n}
            {\vdashc[\Theta n]  \get r : G}
        \end{align*}
        By small step semantics we have $\forall e. \smstep[\Theta][A][n](\get r,e)=\delta(\return{e(r)},e)$ and we can conclude. 
    \end{itemize}
\end{proof}
\section{Soundness}
\label{app:soundness} 
\begin{figure}[p]
    \begin{minipage}{1.2\textwidth}

     \begin{mathpar}\mprset{sep=0.2em,  vskip =0.01ex}

        \polyStarInferRule{\inferrule* [right=true]{~}{\vdashv[\Theta] \true : \mathbb{B}}}
        := 0

        \polyStarInferRule{\inferrule* [right=false]{~}{\vdashv[\Theta] \false : \mathbb{B}}}
        := 0

        \polyStarInferRule{\inferrule* [right=unit]{~}{\vdashv[\Theta] * : \mathbb{U}}}
        := 0

        \polyStarInferRule{\inferrule* [right=string]{
        s \in \{0,1\}^c \\p : i \mapsto c \text{ is a const. polynomial}}{\vdashv[\Theta] s : \mathbb{S}[p]}}
        := 0

        \polyStarInferRule{\inferrule* [right=tensor]{
        \pi_1 \triangleright \Gamma \vdashv[\Theta] Z_1 : P \\
        \pi_2 \triangleright \Delta \vdashv[\Theta] Z_2 : Q
        }{\Gamma \boxplus \Delta \vdashv[\Theta] \lis{Z_1, Z_2} : P \otimes Q}}
        := \polyStar{\pi_1} + \polyStar{\pi_2}

        \polyStarInferRule{\inferrule* [right=bang]{\pi\triangleright \Gamma \vdashc[\Theta] M : A}{p \ast \Gamma \vdashv[\Theta] \oc M : \oc_p A}}
        := p \times \poly{\pi}+p

%
    \end{mathpar}

\caption{The $\polyStar{\cdot}$ function}
\label{fig:poly_sub-cost}
\end{minipage}%
\end{figure}

The proof of substitution in the context of polynomial cost assignment and typing relies on several key auxiliary lemmas and definitions that establish fundamental properties.

First of all, we need a variation on the theme of the function $\poly{\cdot}$, which is given in Figure \ref{fig:poly_sub-cost}, and which only applies to values. By an easy inspection, one realizes that $\polyStar{\pi}\leq\poly{\pi}$ whenever the former is defined. 

The following lemmas establish key foundational properties about values of ground types. Lemma~\ref{lemma:typableValuesGroudType} ensures that every value in the semantic interpretation of a ground type is typable by a derivation in an empty reference context. Lemma~\ref{lemma:costValueGroundType} shows that the cost assigned by $\poly{\cdot}$ to the derivation with reference context $\Theta$ for a closed value of ground type matches exactly the size of the ground type, reflecting their minimal complexity. Finally, Lemma~\ref{lemma:costSubValueGroundType} shows that the refined cost measure $\polyStar{\cdot}$ assigns zero cost to such derivations. 

\begin{lemma}\label{lemma:typableValuesGroudType}
For any ground type $G$, $\forall n\geq 1.\forall V\in \den{G}_n$ then $\vdashv[\varnothing] V:Gn$.
\end{lemma}
\begin{sproof}
    By cases inspection on $G$.
\end{sproof}

\begin{lemma}\label{lemma:costValueGroundType}
For any ground type $G$ and value $V$ with derivation $\pi\triangleright \vdashv[\Theta] V:G$, we have that $\poly{\pi}=\sizeG{G}$.
\end{lemma}
\begin{sproof}
    By induction on the last rule of $\pi$. 
\end{sproof}

\begin{lemma} \label{lemma:costSubValueGroundType}
For any ground type $G$ and value $V$ with derivation $\pi\triangleright \Gamma \vdashv[\Theta] V:G$, we have that $\polyStar{\pi}=0$.
\end{lemma}
\begin{sproof}
    By induction on the last rule of $\pi$. 
\end{sproof}

Moreover, Lemma~\ref{lemma:costSubOperationPositiveTypes} ensures that the polynomial cost assignment $\polyStar{\cdot}$ is coherent with the structure of positive types built using the binary operation $P \boxplus Q$ and the unary operation $p\ast P$ on positive types. In both cases, the lemma formally establishes that the polynomial assignment $\polyStar{\cdot}$ aligns with the inductive structure of positive types. This alignment is essential for reasoning modularly about cost during substitution and composition of terms.

\begin{lemma} \label{lemma:costSubOperationPositiveTypes}\leavevmode\vspace{-2mm}
\begin{enumerate}[label=\Alph*.]
    \item For a value $V$, positive types $P$ and $Q$, if $\pi\triangleright  \vdashv[\Theta] V:P \boxplus Q$, then there exist derivations $\pi_1\triangleright \vdashv[\Theta] V:P$ and $\pi_2\triangleright  \vdashv[\Theta] V:Q$ such that $\polyStar{\pi}=\polyStar{\pi_1}+\polyStar{\pi_2}$.
    \item For a positive value $Z$, positive type $P$ and polynomial $p$ in $\N_{\geq 1}[i]$, if $\pi\triangleright \vdashv[\Theta] Z:p\ast P$, then there exists a derivation $\pi'\triangleright  \vdashv[\Theta] Z: P$ such that $\polyStar{\pi}=p\times\polyStar{\pi'}$.
\end{enumerate}
\end{lemma} 
\begin{sproof}
    The first case is proved by induction on the definition of $P\boxplus Q$ and the second case is by induction on the definition of $p\ast P$. 
\end{sproof}

Building on these foundational cost decompositions, we now state the substitution lemma. It ensures that substituting a well-typed value for a variable within a term or another value preserves typing and adds at most the cost of the substituted value to the overall cost. As described earlier, this additional cost is measured using the refined cost function $\polyStar{\cdot}$, which provides a tighter bound on the complexity contribution of values. 
\begin{lemma}[Substitution] \label{lemma:substitution_lamBLL}
    \leavevmode\vspace{-2mm}
    \begin{enumerate}
        \item Given two derivations $\pi\triangleright \Gamma,x:P \vdashc[\Theta] M:A$ and $\rho\mbox{ }\triangleright \vdashv[\Theta] V:P$, then there exists a derivation $\pi[\rho/x]\triangleright \Gamma\vdashc[\Theta] M[V/x]:A$ such that 
        $\poly{\pi[\rho/x]}\leq \poly{\pi}+\polyStar{\rho}$.
        \item Given two derivations $\varphi \triangleright \Gamma,x:P \vdashv[\Theta] U:A$ and $\rho\mbox{ }\triangleright \vdashv[\Theta] V:P$, then there exists a derivation $\varphi[\rho/x]\triangleright \Gamma\vdashc[\Theta] U[V/x]:A$ such that 
        $\poly{\varphi[\rho/x]}\leq \poly{\varphi}+\polyStar{\rho}$.
    \end{enumerate}
\end{lemma} 
\begin{proof}
    By mutual induction on $\pi$ and $\varphi$.
    \begin{itemize}
        \item \textit{Case $\varphi \triangleright x:P \vdashv[\Theta] x:P$}:\\
        By hypothesis we have $\rho\mbox{ }\triangleright \vdashv[\Theta] V:P$, then we have a derivation $\varphi[\rho/x]\triangleright \vdashv[\Theta] V:P$. Moreover, by definitions of $\mathit{poly}$ and $\mathit{\underline{poly}}$ we have $\poly{\varphi[\rho/x]}=\sizeP{P}=\polyStar{\rho}$. Therefore, we can conclude as we have that
        $\poly{\varphi[\rho/x]}=\sizeP{P}=\polyStar{\rho} \leq \poly{\varphi}+\polyStar{\rho} $. \\

        \item \textit{Case $\varphi \triangleright p \ast \Gamma \vdashv[\Theta] !N:!_p A$}:\\
        We have that $p\ast\Gamma=p\ast\Delta,x:P$, so we have that $P=p\ast Q$ for some positive type $Q$. By BANG-typing rule we have $\varphi_1 \triangleright \Delta,x:Q \vdashc[\Theta] N: A$ and $\poly{\varphi}=p\times\poly{\varphi_1}$.
        By hypothesis we have $\rho\mbox{ }\triangleright \vdashv[\Theta] V:p\ast Q$ and by Lemma~\ref{lemma:costSubOperationPositiveTypes} case B, we have that there exists $\rho'\mbox{ }\triangleright \vdashv[\Theta] V:Q$ such that 
        \begin{equation}
            \label{eq:subst_bang}
            \polyStar{\rho}=p\times\polyStar{\rho'}
        \end{equation}
        By IH on $\varphi_1$ we have that there exists $\varphi[\rho'/x]\triangleright \Delta \vdashc[\Theta] N[V/x]:A$ such that $\poly{\varphi[\rho'/x]}\leq \poly{\varphi_1}+\polyStar{\rho'}$.\\
        Therefore, let $\varphi[\rho/x]\triangleright p\ast\Delta \vdashv[\Theta] ! N[V/x]:!_p A$, we have that
        \begin{align*}
            \poly{\varphi[\rho/x]} 
            &= p \times \poly{\varphi[\rho'/x]}\\
            &\hspace{-1mm}\overset{IH}{\leq} p \times (\poly{\varphi_1}+\polyStar{\rho'})\\
            &= p \times \poly{\varphi_1}+ p \times \polyStar{\rho'}\\
            &\hspace{-3mm}\overset{\text{Eq.}\ref{eq:subst_bang}}{=} p \times \polyStar{\varphi_1}+\poly{\rho}\\
            &=\poly{\varphi}+\polyStar{\rho}
        \end{align*}
        and we can conclude.\\

        \item \textit{Case $\varphi \triangleright \Gamma_1 \boxplus \Gamma_2 \vdashv[\Theta] \lis{Z_1,Z_2}:Q_1 \otimes Q_2$}:\\
        We have that $\Gamma_1 \boxplus \Gamma_2 = \Delta_1 \boxplus \Delta_2, x:P$ for some $\Delta_1, \Delta_2$ and by hypothesis we have $\rho\mbox{ }\triangleright \vdashv[\Theta] V:P$. We proceed by analyzing the following sub-cases:
        \begin{itemize}
            \item Sub-case $x\in\Gamma_1$ and $x\notin \Gamma_2$:\\
            We have that $\Gamma_1=\Delta_1,x:P$ and $\Gamma_2=\Delta_2$. By TENSOR-typing  rule we have that $\varphi_1 \triangleright \Delta_1,x:P \vdashv[\Theta] Z_1: Q_1$ and $\varphi_2 \triangleright \Delta_2 \vdashv[\Theta] Z_2: Q_2$. Moreover, by definition and $\poly{\varphi}=1+\poly{\varphi_1}+\poly{\varphi_2}$.\\
            By IH on $\varphi_1$ we have that there exists $\varphi_1[\rho/x] \triangleright \Delta_1 \vdashv[\Theta] Z_1[V/x]: Q_1$ such that $\poly{\varphi_1[\rho/x]}\leq \poly{\varphi_1}+\polyStar{\rho}$.\\
            Therefore, let $\varphi[\rho/x]\triangleright \Delta_1 \boxplus \Delta_2 \vdashv[\Theta] \lis{Z_1[V/x],Z_2}:Q_1 \otimes Q_2$, we have that
            \begin{align*}
                \poly{\varphi[\rho/x]} 
                &= 1+\poly{\varphi_1[\rho/x]} + \poly{\varphi_2} \\
                &\hspace{-1mm}\overset{IH}{\leq} 1+ \poly{\varphi_1} + \polyStar{\rho} + \poly{\varphi_2}\\ 
                &=\poly{\varphi}+\polyStar{\rho}
            \end{align*}
            and we can conclude.
            \item Sub-case $x\notin\Gamma_1$ and $x\in\Gamma_2$: similar to the previous sub-case.
            \item Sub-case $P=P_1\boxplus P_2$:\\
            We have that $\Gamma_1=\Delta_1,x:P_1$ and $\Gamma_2=\Delta_2, x:P_2$. By TENSOR-typing  rule we have that $\varphi_1 \triangleright \Delta_1,x:P_1 \vdashv[\Theta] Z_1: Q_1$ and $\varphi_2 \triangleright \Delta_2,x:P_2 \vdashv[\Theta] Z_2: Q_2$. Moreover, by definition and $\poly{\varphi}=1+\poly{\varphi_1}+\poly{\varphi_2}$.\\
            By hypothesis we have $\rho\mbox{ }\triangleright \vdashv[\Theta] V:P$ and by Lemma~\ref{lemma:costSubOperationPositiveTypes} case A, we have that there exists $\rho_1\mbox{ }\triangleright \vdashv[\Theta] V:P_1$ and $\rho_2\mbox{ }\triangleright \vdashv[\Theta] V:P_2$ such that 
            \begin{equation}
                \label{eq:subst_tensor}
                \polyStar{\rho}=\polyStar{\rho_1}+\polyStar{\rho_2}
            \end{equation}
            By IH on $\varphi_1$ we have that there exists $\varphi_1[\rho_1/x] \triangleright \Delta_1 \vdashv[\Theta] Z_1[V/x]: Q_1$ such that $\poly{\varphi_1[\rho_1/x]}\leq \poly{\varphi_1}+\polyStar{\rho_1}$.\\
            By IH on $\varphi_2$ we have that there exists $\varphi_2[\rho_2/x] \triangleright \Delta_2 \vdashv[\Theta] Z_2[V/x]: Q_2$ such that $\poly{\varphi_2[\rho_2/x]}\leq \poly{\varphi_2}+\polyStar{\rho_2}$.\\
            Therefore, let $\varphi[\rho/x]\triangleright \Delta_1\boxplus\Delta_2 \vdashv[\Theta] \lis{Z_1[V/x],Z_2[V/x]}:Q_1\otimes Q_2$, we have that
            \begin{align*}
                \poly{\varphi[\rho/x]} 
                &= 1+\poly{\varphi_1[\rho_1/x]} + \poly{\varphi_2[\rho_2/x]} \\
                &\hspace{-1mm}\overset{IHs}{\leq} 1+\poly{\varphi_1} + \polyStar{\rho_1} + \poly{\varphi_2} +\polyStar{\rho_2}\\ 
                &\hspace{-2mm}\overset{\text{Eq.}\ref{eq:subst_tensor}}{=} 1+\poly{\varphi_1} + \poly{\varphi_2} +\polyStar{\rho}\\
                &=\poly{\varphi}+\polyStar{\rho}
            \end{align*}
            and we can conclude.\\
        \end{itemize}
        
        \item \textit{Case $\varphi \triangleright \Gamma,x:P \vdashv[\Theta] \lambda y.N:P'\multimap A'$}:\\
        By LAM-typing rule we have $\varphi_1 \triangleright \Gamma,x:P,y:P'\vdashc[\Theta] N:A'$ and $\poly{\varphi}=1+\poly{\varphi_1}$. Moreover, by hypothesis we have $\rho\mbox{ }\triangleright \vdashv[\Theta] V:P$.\\
        By IH on $\varphi_1$ we have that there exists $\varphi_1[\rho/x] \triangleright \Gamma,y:P' \vdashc[\Theta] N[V/x]:A'$ such that $\poly{\varphi_1[\rho/x]}\leq \poly{\varphi_1}+\polyStar{\rho}$.\\
        Therefore, let $\varphi[\rho/x]\triangleright \Gamma \vdashv[\Theta] \lambda y.N[V/x]: P'\multimap A'$, we have that
        \begin{align*}
            \poly{\varphi[\rho/x]} 
            &= 1+\poly{\varphi_1[\rho/x]}\\
            &\hspace{-1mm}\overset{IH}{\leq} 1+\poly{\varphi_1} + \polyStar{\rho} \\ 
            &=\poly{\varphi}+\polyStar{\rho}
        \end{align*}
        and we can conclude.\\

        \item \textit{Case $\pi \triangleright \Gamma,x:P \vdashc[\Theta] \return{V'}:A$}:\\
        By ETA-typing rule we have $\pi_1 \triangleright \Gamma,x:P \vdashv[\Theta] V':A$ and $\poly{\pi}=1+\poly{\pi_1}$. Moreover, by hypothesis we have $\rho\mbox{ }\triangleright \vdashv[\Theta] V:P$.\\
        By IH on $\pi_1$ we have that there exists $\pi_1[\rho/x] \triangleright \Gamma \vdashv[\Theta] V'[V/x]:A$ such that $\poly{\pi_1[\rho/x]}\leq \poly{\pi_1}+\polyStar{\rho}$.\\
        Therefore, let $\pi[\rho/x]\triangleright \Gamma \vdashc[\Theta] \return{V'[V/x]}:A$, we have that
        \begin{align*}
            \poly{\pi[\rho/x]} 
            &= 1+\poly{\pi_1[\rho/x]} \\
            &\hspace{-1mm}\overset{IH}{\leq} 1+\poly{\pi_1}+\polyStar{\rho}\\ 
            &=\poly{\pi}+\polyStar{\rho}
        \end{align*}
        and we can conclude.\\

        \item \textit{Case $\pi \triangleright \Gamma,x:P \vdashc[\Theta] \der(Z):A$}:\\
        By DER-typing rule we have $\pi_1 \triangleright \Gamma,x:P \vdashv[\Theta] Z:!_1 A$ and $\poly{\pi}=1+\poly{\pi_1}$. Moreover, by hypothesis we have $\rho\mbox{ }\triangleright \vdashv[\Theta] V:P$.\\
        By IH on $\pi_1$ we have that there exists $\pi_1[\rho/x] \triangleright \Gamma \vdashv[\Theta] Z[V/x]:!_1 A$ such that $\poly{\pi_1[\rho/x]}\leq \poly{\pi_1}+\polyStar{\rho}$.\\
        Therefore, let $\pi[\rho/x]\triangleright \Gamma \vdashc[\Theta] \der(Z[V/x]):A$, we have that
        \begin{align*}
            \poly{\pi[\rho/x]} 
            &=1+\poly{\pi_1[\rho/x]} \\
            &\hspace{-1mm}\overset{IH}{\leq} 1+\poly{\pi_1}+\polyStar{\rho}\\ 
            &=\poly{\pi}+\polyStar{\rho}
        \end{align*}
        and we can conclude.\\

        \item \textit{Case $\pi \triangleright \Gamma,x:P \vdashc[\Theta] \setr{r}{Z}:\UU$}:\\
        By SET-typing rule we have $\pi_1 \triangleright \Gamma,x:P \vdashv[\Theta] Z:G$ where $r:G\in\Theta$ and $\poly{\pi}=2+\poly{\pi_1}$. Moreover, by hypothesis we have $\rho\mbox{ }\triangleright \vdashv[\Theta] V:P$.\\
        By IH on $\pi_1$ we have that there exists $\pi_1[\rho/x] \triangleright \Gamma \vdashv[\Theta] Z[V/x]:G$ such that $\poly{\pi_1[\rho/x]}\leq \poly{\pi_1}+\polyStar{\rho}$. \\
        Therefore, let $\pi[\rho/x]\triangleright \Gamma \vdashc[\Theta] \setr{r}{Z[V/x]}:\UU$, we have that
        \begin{align*}
            \poly{\pi[\rho/x]} 
            &=2+\poly{\pi_1[\rho/x]} \\
            &\hspace{-1mm}\overset{IH}{\leq} 2+ \poly{\pi}+\polyStar{\rho}\\ 
            &=2+\poly{\pi_1}+\polyStar{\rho}
        \end{align*}
        and we can conclude.\\

        \item \textit{Case $\pi \triangleright \Gamma_1 p \boxplus \Gamma_2 p \vdashc[\Theta] f_p(Z_1,Z_2):Gp$}:\\
        We have that $\Gamma_1 p\boxplus \Gamma_2  p = \Delta_1  p\boxplus \Delta_2  p, x:P$ for some $\Delta_1, \Delta_2$ and by hypothesis we have $\rho\mbox{ }\triangleright \vdashv[\Theta] V:P$. We proceed by analyzing the following sub-cases:
        \begin{itemize}
            \item Sub-case $x\in\Gamma_1 p$ and $x\notin \Gamma_2 p$:\\
            We have that $\Gamma_1 p=\Delta_1 p ,x:P$ and $\Gamma_2 p=\Delta_2 p$. By FUN-typing rule $\typeof{f}:G_1\times G_2 \rightarrow G$ and we have $\pi_1 \triangleright \Delta_1 p,x:P \vdashv[\Theta] Z_1:G_1p$ and $\pi_2 \triangleright \Delta_2 p \vdashv[\Theta] Z_2:G_2p$. Moreover, by definition we have $\poly{\pi}=1+\poly{\pi_1}+\poly{\pi_2}+\sizeG{G}$.\\
            By IH on $\pi_1$ we have that there exists $\pi_1[\rho/x] \triangleright \Delta_1 p  \vdashv[\Theta] Z_1[V/x]:G_1p$ such that $\poly{\pi_1[\rho/x]}\leq \poly{\pi_1}+\polyStar{\rho}$.  \\
            Therefore, let $\pi[\rho/x]\triangleright \Delta_1 p \boxplus \Delta_2 p \vdashc[\Theta] f_p(Z_1[V/x],Z_2):Gp$, we have that
            \begin{align*}
                \poly{\pi[\rho/x]} 
                &=1+\poly{\pi_1[\rho/x]}+\poly{\pi_2}+\sizeG{G}\\
                &\hspace{-2mm}\overset{\text{IH}}{\leq}  
                1+\poly{\pi_1}+\polyStar{\rho}+\poly{\pi_2}+\sizeG{G}\\
                &=\poly{\pi}+\polyStar{\rho} 
            \end{align*}
            and we can conclude.
            \item Sub-case $x\notin\Gamma_1 p$ and $x\in \Gamma_2 p$: similar to the previous sub-case.
            \item Sub-case $P=P_1\boxplus P_2$:\\
            We have that $\Gamma_1 p=\Delta_1 p ,x:P_1$ and $\Gamma_2 p=\Delta_2 p, x:P_2$. By FUN-typing rule $\typeof{f}:G_1\times G_2 \rightarrow G$ and we have $\pi_1 \triangleright \Delta_1 p,x:P \vdashv[\Theta] Z_1:G_1p$ and $\pi_2 \triangleright \Delta_2 p, x:P_2 \vdashv[\Theta] Z_2:G_2p$. Moreover, by definition we have $\poly{\pi}=1+\poly{\pi_1}+\poly{\pi_2}+\sizeG{G}$.\\
            By hypothesis we have $\rho\mbox{ }\triangleright \vdashv[\Theta] V:P$ and by Lemma~\ref{lemma:costSubOperationPositiveTypes} case A we have that there exists $\rho_1\mbox{ }\triangleright \vdashv[\Theta] V:P_1$ and $\rho_2\mbox{ }\triangleright \vdashv[\Theta] V:P_2$ such that $\polyStar{\rho}=\polyStar{\rho_1}+\polyStar{\rho_2}$.\\
            By IH on $\pi_1$ we have that there exists $\pi_1[\rho_1/x] \triangleright \Delta_1 p  \vdashv[\Theta] Z_1[V/x]:G_1p$ such that $\poly{\pi_1[\rho_1/x]}\leq \poly{\pi_1}+\polyStar{\rho_1}$.\\ 
            By IH on $\pi_2$ we have that there exists $\pi_2[\rho_2/x] \triangleright \Delta_2 p  \vdashv[\Theta] Z_2[V/x]:G_2p$ such that $\poly{\pi_2[\rho_2/x]}\leq \poly{\pi_2}+\polyStar{\rho_2}$. \\ 
            Therefore, let $\pi[\rho/x]\triangleright \Delta_1 p \boxplus \Delta_2 p \vdashc[\Theta] f_p(Z_1[V/x],Z_2[V/x]):Gp$, we have that
            \begin{align*}
                \poly{\pi[\rho/x]} 
                &=1+\poly{\pi_1[\rho_1/x]}+\poly{\pi_2[\rho_2/x]}+\sizeG{G}\\
                &\hspace{-1mm}\overset{\text{IHs}}{\leq}  1+\poly{\pi_1}+\polyStar{\rho_1}+\poly{\pi_2}+\polyStar{\rho_2}+\sizeG{G}\\ 
                &=1+\poly{\pi_1}+\poly{\pi_2}+\sizeG{G}+\polyStar{\rho}\\
                &=\poly{\pi}+\polyStar{\rho}
            \end{align*}
            and we can conclude.\\
        \end{itemize}

        \item \textit{Case $\pi \triangleright \Gamma_1 \boxplus \Gamma_2 \vdashc[\Theta] \letin{\lis{y_1,y_2}=Z}{N}:A$}:\\
        We have that $\Gamma_1 \boxplus \Gamma_2 = \Delta_1 \boxplus \Delta_2, x:P$ for some $\Delta_1, \Delta_2$ and by hypothesis we have $\rho\mbox{ }\triangleright \vdashv[\Theta] V:P$. We proceed by analyzing the following sub-cases:
        \begin{itemize}
            \item Sub-case $x\in\Gamma_1$ and $x\notin\Gamma_2$:\\
            We have that $\Gamma_1=\Delta_1,x:P$ and $\Gamma_2=\Delta_2$.
            By LET\_PAIR-typing rule we have $\pi_1 \triangleright \Delta_1, x:P \vdashv[\Theta] Z:Q_1 \tensor Q_2$ and $\pi_2 \triangleright \Delta_2 \vdashc[\Theta] N:A$. Moreover, by definition we have $\poly{\pi}=1+\poly{\pi_1}+\poly{\pi_2}$.\\
            By IH on $\pi_1$ we have that there exists $\pi_1[\rho/x] \triangleright \Delta_1  \vdashv[\Theta] Z[V/x]:Q_1 \tensor Q_2$ such that $\poly{\pi_1[\rho/x]}\leq \poly{\pi_1}+\polyStar{\rho}$. \\
            Therefore, let $\pi[\rho/x]\triangleright \Delta_1 \boxplus \Delta_2  \vdashc[\Theta] \letin{\lis{y_1,y_2}=Z[V/x]}{N}:A$, we have that
            \begin{align*}
                \poly{\pi[\rho/x]} 
                &=1+\poly{\pi_1[\rho/x]}+\poly{\pi_2} \\
                &\hspace{-1mm}\overset{IH}{\leq} 1+\poly{\pi_1}+\polyStar{\rho}+\poly{\pi_2}\\ 
                &=\poly{\pi}+\polyStar{\rho}
            \end{align*}
            and we can conclude.
            \item Sub-case $x\notin\Gamma_1$ and $x\in\Gamma_2$:\\
            We have that $\Gamma_2=\Delta_2,x:P$ and $\Gamma_1=\Delta_1$.
            By LET\_PAIR-typing rule we have $\pi_1 \triangleright \Delta_1, x:P \vdashv[\Theta] Z:Q_1 \tensor Q_2$ and $\pi_2 \triangleright \Delta_2 \vdashc[\Theta] N:A$. Moreover, by definition we have $\poly{\pi}=1+\poly{\pi_1}+\poly{\pi_2}$.\\
            By IH on $\pi_2$ we have that there exists $\pi_2[\rho/x] \triangleright \Delta_2  \vdashv[\Theta] N[V/x]:A$ such that $\poly{\pi_2[\rho/x]}\leq \poly{\pi_2}+\polyStar{\rho}$. \\
            Therefore, let $\pi[\rho/x]\triangleright \Delta_1 \boxplus \Delta_2  \vdashc[\Theta] \letin{\lis{y_1,y_2}=Z}{N[V/x]}:A$, we have that
            \begin{align*}
                \poly{\pi[\rho/x]} 
                &=1+\poly{\pi_1}+\poly{\pi_2[\rho/x]} \\
                &\hspace{-1mm}\overset{IH}{\leq} 1+\poly{\pi_1}+\poly{\pi_2}+\polyStar{\rho}\\ 
                &=\poly{\pi}+\polyStar{\rho}
            \end{align*}
            and we can conclude.
            \item Sub-case $P=P_1\boxplus P_2$:\\
            We have that $\Gamma_1=\Delta_1 ,x:P_1$ and $\Gamma_2=\Delta_2, x:P_2$.  By LET\_PAIR-typing rule we have $\pi_1 \triangleright \Delta_1, x:P_1 \vdashv[\Theta] Z:Q_1 \tensor Q_2$ and $\pi_2 \triangleright \Delta_2, x:P_2 \vdashc[\Theta] N:A$. Moreover, by definition we have $\poly{\pi}=1+\poly{\pi_1}+\poly{\pi_2}$.\\ 
            By hypothesis we have $\rho\mbox{ }\triangleright \vdashv[\Theta] V:P$ and by Lemma~\ref{lemma:costSubOperationPositiveTypes} case A we have that there exists $\rho_1\mbox{ }\triangleright \vdashv[\Theta] V:P_1$ and $\rho_2\mbox{ }\triangleright \vdashv[\Theta] V:P_2$ such that 
            \begin{equation}
                \label{eq:subst_let_pair}
                \polyStar{\rho}=\polyStar{\rho_1}+\polyStar{\rho_2}.
            \end{equation}
            By IH on $\pi_1$ we have that there exists $\pi_1[\rho_1/x] \triangleright \Delta_1  \vdashv[\Theta] Z[V/x]:Q_1 \tensor Q_2$ such that $\poly{\pi_1[\rho_1/x]}\leq \poly{\pi_1}+\polyStar{\rho_1}$. \\
            By IH on $\pi_2$ we have that there exists $\pi_2[\rho_2/x] \triangleright \Delta_2  \vdashc[\Theta] N[V/x]:A$ such that $\poly{\pi_1[\rho_2/x]}\leq \poly{\pi_2}+\polyStar{\rho_2}$.\\
            Therefore, let $\pi[\rho/x]\triangleright \Delta_1 \boxplus \Delta_2  \vdashc[\Theta] \letin{\lis{y_1,y_2}=Z[V/x]}{N[V/x]}:A$, we have that
            \begin{align*}
                \poly{\pi[\rho/x]} 
                &=1+\poly{\pi_1[\rho_1/x]}+\poly{\pi_2[\rho_2/x]} \\
                &\hspace{-1mm}\overset{IHs}{\leq} 1+\poly{\pi_1}+\polyStar{\rho_1}+\poly{\pi_2}+\polyStar{\rho_2}\\ 
                &\hspace{-2mm}\overset{\text{Eq.}\ref{eq:subst_let_pair}}{=} 
                1+\poly{\pi_1}+\poly{\pi_2}+\polyStar{\rho}\\ 
                &=\poly{\pi}+\polyStar{\rho}
            \end{align*}
            and we can conclude.\\
        \end{itemize}

        \item \textit{Case $\pi \triangleright \Gamma_1 \boxplus \Gamma_2 \vdashc[\Theta] \letin{y=Z}{N}:A$ and case $\pi \triangleright \Gamma_1 \boxplus \Gamma_2 \vdashc[\Theta] NZ:A$}: similar to the previous case.\\
        
        \item \textit{Case $\pi \triangleright (p\ast\Gamma_1) \boxplus \Gamma_2 \vdashc[\Theta] \loopp{V'}{p}{N}:Q$}:\\
        We have that $(p\ast\Gamma_1) \boxplus \Gamma_2=(p\ast\Delta_1)\boxplus\Delta_2,x:P$ for some $\Delta_1,\Delta_2$ and by hypothesis we have $\rho\mbox{ }\triangleright \vdashv[\Theta] V:P$. We proceed by analyzing the following sub-cases:
        \begin{itemize}
            \item Sub-case $x\in\Gamma_1$ and $x\notin\Gamma_2$:\\
            Since $x\in\Gamma_1$, we have that $P=p\ast P'$ for some positive type $P'$. Therefore, we have $\Gamma_1=\Delta_1,x:P'$ and $\Gamma_1=\Delta_2$. By LOOP-typing rule we have $\pi_1 \triangleright \Delta_1,x:P' \vdashv[\Theta] V':Q \multimap Q$ and $\pi_2 \triangleright \Delta_2 \vdashc[\Theta]N:Q$. Moreover, by definition we have that $\poly{\pi}=1+p+p\times \poly{\pi_1}+\poly{\pi_2}$.\\
            By hypothesis we have $\rho\mbox{ }\triangleright \vdashv[\Theta] V:p\ast P'$ and by Lemma~\ref{lemma:costSubOperationPositiveTypes} case B we have that there exists $\rho'\mbox{ }\triangleright \vdashv[\Theta] V:P'$ such that 
            \begin{equation}
                \label{eq:subst_loop_1}
                \polyStar{\rho}=p\times\polyStar{\rho'}.
            \end{equation}
            By IH on $\pi_1$ we have that there exists $\pi_1[\rho'/x] \triangleright \Delta_1  \vdashv[\Theta] V'[V/x]:Q\multimap Q$ such that $\poly{\pi_1[\rho'/x]}\leq \poly{\pi_1}+\polyStar{\rho'}$. \\
            Therefore, let $\pi[\rho/x]\triangleright (p\ast\Delta_1) \boxplus \Delta_2  \vdashc[\Theta] \loopp{(V'[V/x])}{p}{N}:Q$, we have that
            \begin{align*}
                \poly{\pi[\rho/x]} 
                &=1+p+p\times \poly{\pi_1[\rho'/x]}+\poly{\pi_2} \\
                &\hspace{-1mm}\overset{IH}{\leq} 1+p+p\times (\poly{\pi_1}+\polyStar{\rho'})+\poly{\pi_2}\\ 
                &=1+p+p\times\poly{\pi_1}+p\times\polyStar{\rho'}+\poly{\pi_2}\\ 
                &\hspace{-2mm}\overset{\text{Eq.}\ref{eq:subst_loop_1}}{=} 
                1+p+p\times\poly{\pi_1}+\polyStar{\rho}+\poly{\pi_2}\\  
                &=\poly{\pi}+\polyStar{\rho}
            \end{align*}
            and we can conclude.
            \item Sub-case $x\notin\Gamma_1$ and $x\in\Gamma_2$:\\
            We have that $\Gamma_1=\Delta_1$ and $\Gamma_2=\Delta_2,x:P$.
            By LOOP-typing rule we have $\pi_1 \triangleright \Delta_1 \vdashv[\Theta] V':Q \multimap Q$ and $\pi_2 \triangleright \Delta_2,x:P' \vdashc[\Theta]N:Q$. Moreover, by definition we have that $\poly{\pi}=1+p+p\times \poly{\pi_1}+\poly{\pi_2}$.\\
            By IH on $\pi_2$ we have that there exists $\pi_2[\rho/x] \triangleright \Delta_2  \vdashc[\Theta]N[V/x]:Q$ such that $\poly{\pi_2[\rho/x]}\leq \poly{\pi_2}+\polyStar{\rho}$.\\
            Therefore, let $\pi[\rho/x]\triangleright (p\ast\Delta_1) \boxplus \Delta_2  \vdashc[\Theta] \loopp{V'}{p}{(N[V/x])}:Q$, we have that
            \begin{align*}
                \poly{\pi[\rho/x]} 
                &=1+p+p\times \poly{\pi_1}+\poly{\pi_2[\rho/x]} \\
                &\hspace{-1mm}\overset{IH}{\leq} 1+p+p\times \poly{\pi_1}+\poly{\pi_2}+\polyStar{\rho}\\  
                &=\poly{\pi}+\polyStar{\rho}
            \end{align*}
            and we can conclude.
            \item Sub-case $x\in\Gamma_1$ and $x\in\Gamma_2$:\\
            We have that $\Gamma_1=\Delta_1,x:P_1$ and $\Gamma_2=\Delta_2,x:P_2$, and $P=(p\ast P_1) \boxplus P_2$.  By LOOP-typing rule we have $\pi_1 \triangleright \Delta_1,x:P_1 \vdashv[\Theta] V':Q \multimap Q$ and $\pi_2 \triangleright \Delta_2,x:P_2 \vdashc[\Theta]N:Q$. Moreover, by definition we have that $\poly{\pi}=1+p+p\times \poly{\pi_1}+\poly{\pi_2}$.\\
            In this case case by hypothesis we have $\rho\mbox{ }\triangleright \vdashv[\Theta] V:(p\ast P_1) \boxplus P_2$.
            By Lemma~\ref{lemma:costSubOperationPositiveTypes} case A on $\rho$ we have that there exists $\rho_1\mbox{ }\triangleright \vdashv[\Theta] V:p\ast P_1$ and $\rho_2\mbox{ }\triangleright \vdashv[\Theta] V: P_2$ such that $\polyStar{\rho}=\polyStar{\rho_1}+\polyStar{\rho_2}$.
            Moreover, by Lemma~\ref{lemma:costSubOperationPositiveTypes} case B on $\rho_1$ we have that there exists $\rho'_1\mbox{ }\triangleright \vdashv[\Theta] V:  P_1$ such that $\polyStar{\rho_1}=p\times\polyStar{\rho'_1} $.
            Summing up, we have that 
            \begin{equation}
                \label{eq:subst_loop_2}
                \polyStar{\rho}=p\times\polyStar{\rho'_1}+\polyStar{\rho_2}.
            \end{equation}
            By IH on $\pi_1$ we have that there exists $\pi_1[\rho'_1/x] \triangleright \Delta_1  \vdashv[\Theta] V'[V/x]:Q\multimap Q$ such that $\poly{\pi_1[\rho'_1/x]}\leq \poly{\pi_1}+\polyStar{\rho'_1}$. \\
            By IH on $\pi_2$ we have that there exists $\pi_2[\rho_2/x] \triangleright \Delta_2  \vdashc[\Theta]N[V/x]:Q$ such that $\poly{\pi_2[\rho_2/x]}\leq \poly{\pi_2}+\polyStar{\rho_2}$.\\
            Therefore, let $\pi[\rho/x]\triangleright (p\ast\Delta_1) \boxplus \Delta_2  \vdashc[\Theta] \loopp{(V'[V/x])}{p}{(N[V/x])}:Q$, we have that
             \begin{align*}
                \poly{\pi[\rho/x]} 
                &=1+p+p\times \poly{\pi_2[\rho'_1/x]}+\polyStar{\pi_2[\rho_2/x]} \\
                &\hspace{-1mm}\overset{IHs}{\leq} 1+p+p\times (\poly{\pi_1}+\poly{\rho'_1})+\poly{\pi_2}+\polyStar{\rho_2}\\ 
                &=1+p+p\times\poly{\pi_1}+p\times\polyStar{\rho'_1}+\poly{\pi_2}+\polyStar{\rho_2}\\  
                &=\poly{\pi}+p\times\polyStar{\rho'_1}+\polyStar{\rho_2}\\  
                &\hspace{-2mm}\overset{\text{Eq.}\ref{eq:subst_loop_2}}{=} \poly{\pi}+\polyStar{\rho}
            \end{align*}
            and we can conclude.\\
        \end{itemize}

        \item \textit{Case $\pi \triangleright \Gamma_1\boxplus \Gamma_2 \vdashc[\Theta] \ifthen{Z}{N_1}{N_2}:A$}:\\ 
        Recall that in Subsection~\ref{sect:soundcomplete}, we restrict attention to \texttt{if}-constructs where both branches are closed terms, ensuring that their evaluation cost does not depend on the surrounding context and can be analyzed independently. Therefore we have $\Gamma_2$ is the empty environment and  $\Gamma_1 \boxplus \Gamma_2 = \Delta_1, x:P$ for some $\Delta_1$. Moreover, by hypothesis we have $\rho\mbox{ }\triangleright \vdashv[\Theta] V:P$. 
        By CASE-typing rule we have $\pi_g \triangleright \Delta_1,x:P \vdashv[\Theta] Z:\BB$, $\pi_1 \triangleright \Delta_2 \vdashc[\Theta] N_1:A$ and $\pi_2 \triangleright \Delta_2 \vdashc[\Theta] N_2:A$. Moreover, by definition we have $\poly{\pi}=1+\poly{\pi_g}+\poly{\pi_1}+\poly{\pi_2}$\\
        By IH on $\pi_g$ we have that there exists $\pi_g[\rho/x] \triangleright \Delta_1 \vdashv[\Theta] Z[V/x]:\BB$ such that $\poly{\pi_g[\rho/x]}\leq \poly{\pi_g}+\polyStar{\rho}$.\\
        Therefore, let $\pi[\rho/x]\triangleright \Delta_1 \boxplus \Delta_2  \vdashc[\Theta] \ifthen{Z[V/x]}{N_1}{N_2}:A$, we have that
             \begin{align*}
                \poly{\pi[\rho/x]} 
                &=1+\poly{\pi_g[\rho/x]}+\poly{\pi_1}+\poly{\pi_2} \\
                &\hspace{-1mm}\overset{IH}{\leq} 
                1+\poly{\pi_g}+\polyStar{\rho}+\poly{\pi_1}+\poly{\pi_2}\\ 
                &=\poly{\pi}+\poly{\rho}
            \end{align*}
        and we can conclude.
    \end{itemize}
\end{proof}

Another crucial property of the polynomial cost assignment is that it strictly decreases along term reduction. The following lemma formalizes this concept by showing that, for any well-typed term and store, each small-step reduction produces a term whose polynomial cost is strictly smaller than that of the original term.

\begin{lemma}[Measure Decreasing along Reduction] \label{lemma:descreasing}
    Given $\pi\triangleright  \vdashc[\Theta] M:A$ and store $e\in \Store_{\Theta}$ such that 
    $\smstep[\Theta][A][n](M,e)=\sum_{k} a_k \mbox{ } \delta(M_k, e_k)$, then for every $k$ there exists $\rho_k \triangleright  \vdashc[\Theta] M_k:A$ such that $\poly{\pi}>\poly{\rho_k}$.
\end{lemma} 
\begin{proof}
    By induction on the definition of small step reduction (Figure~\ref{fig:smallstep}).
    \begin{itemize}
        \item \textit{Case $\pi\triangleright  \vdashc[\Theta] \letin{x=\return{V}}{N}:A$ and $\smstep[\Theta][A][n](M,e)=\delta(N[V/x],e)$}:\\
        By LET-typing rule we have $\pi_1 \triangleright  \vdashc[\Theta] \return{V}:P$ and $\pi_2 \triangleright x:P \vdashc[\Theta] N:A$ and by definition 
        \begin{equation}
            \label{eq:dec_let_1}
            \poly{\pi}=1+\poly{\pi_1}+\poly{\pi_2}.
        \end{equation}
         Moreover, by ETA-typing rule applied to $\pi_1$ we have $\pi'_1 \triangleright  \vdashv[\Theta] V:P$ and by definition 
        \begin{equation}
            \label{eq:dec_let_2}
            \poly{\pi_1}=1+\poly{\pi'_1}.
        \end{equation} 
        By Lemma~\ref{lemma:substitution_lamBLL} applied to $\pi_2$ and $\pi'_1$, we have that there exists $\pi_2[\pi'_1/x] \triangleright \vdashc[\Theta] N[V/x]:A$ such that 
        $\poly{\pi_2[\pi'_1/x]}\leq\poly{\pi_2}+\polyStar{\pi'_1}$.
        Moreover, we have that $\polyStar{\cdot}$ is always lower or equal to $\poly{\cdot}$, then we have:
        \begin{equation}
            \label{eq:dec_let_4}
            \poly{\pi_2[\pi'_1/x]}\leq\poly{\pi_2}+\poly{\pi'_1}.
        \end{equation} 
        Therefore, we have that 
        \begin{align*}
            \poly{\pi} &\overset{\text{Eq.~\ref{eq:dec_let_1}}}{=} 1+\poly{\pi_1}+\poly{\pi_2}\\
            &\overset{\text{Eq.~\ref{eq:dec_let_2}}}{=} 2+\poly{\pi'_1}+\poly{\pi_2}\\
            &\hspace{2mm} > \poly{\pi_2[\pi'_1/x]}\\
            &\overset{\text{Eq.~\ref{eq:dec_let_4}}}{\leq} \poly{\pi'_1}+\poly{\pi_2}
        \end{align*}
        and we can conclude.\\

        \item \textit{Case $\pi\triangleright  \vdashc[\Theta] NZ:A$ and $\smstep[\Theta][A][n](M,e)=\delta(N_k Z,e_k)$}:\\
        By APP-typing rule we have  $\pi_1 \triangleright  \vdashc[\Theta] N:P\multimap A$ and $\pi_2 \triangleright  \vdashv[\Theta] Z:P$ and by definition 
        \begin{equation}
            \label{eq:dec_app_1}
            \poly{\pi}=1+\poly{\pi_1}+\poly{\pi_2}.
        \end{equation}
        Moreover, by small step semantics $\smstep[\Theta][A][n](N,e)=\sum_{k} a_k \mbox{ }\delta(N_k,e_k)$.\\
        By IH on $\pi_1$ we have that for every $k$ there exists $\rho'_k \triangleright  \vdashc[\Theta] N_k:P\multimap A$ such that $\poly{\pi_1}>\poly{\rho'_k}$.\\
        Let $\rho_k \triangleright  \vdashc[\Theta] N_k Z: A$, then by definition we have 
        \begin{equation}
            \label{eq:dec_app_2}
            \poly{\rho_k}=1+\poly{\rho'_k}+\poly{\pi_2}.
        \end{equation}
        Therefore, we have that
        \begin{align*}
            \poly{\pi} \overset{\text{Eq.~\ref{eq:dec_app_1}}}{=} 1+\poly{\pi_1}+\poly{\pi_2}
            \overset{IH}{>} 1+\poly{\rho'_k}+\poly{\pi_2}
            \overset{\text{Eq.~\ref{eq:dec_app_2}}}{=} \poly{\rho_k}
        \end{align*}
        this holds for every $k$ and so we can conclude.\\

        \item \textit{Case $\pi\triangleright  \vdashc[\Theta] \letin{x=N}{N'}:A$ and $\smstep[\Theta][A][n](M,e)=\sum_{k} a_k \mbox{ } \delta(\letin{x=N_k}{N'}, e_k)$}:\\
        By LET-typing rule we have $\pi_1 \triangleright  \vdashc[\Theta] N:P$ and $\pi_2 \triangleright x:P \vdashc[\Theta] N':A$ and by definition 
        \begin{equation}
            \label{eq:dec_let_2_1}
            \poly{\pi}=1+\poly{\pi_1}+\poly{\pi_2}.
        \end{equation}
        Moreover, by small step semantics $\smstep[\Theta][A][n](N,e)=\sum_{k} a_k \mbox{ }\delta(N_k,e_k)$.\\
        By IH on $\pi_1$ we have that for every $k$ there exists $\rho'_k \triangleright  \vdashc[\Theta] N_k:P$ such that $\poly{\pi_1}>\poly{\rho'_k}$.\\
        Let $\rho_k \triangleright  \vdashc[\Theta] \letin{x=N_k}{N'}: A$, then by definition we have 
        \begin{equation}
            \label{eq:dec_let_2_2}
            \poly{\rho_k}=1+\poly{\rho'_k}+\poly{\pi_2}.
        \end{equation}
        Therefore, we have that
        \begin{align*}
            \poly{\pi} \overset{\text{Eq.~\ref{eq:dec_let_2_1}}}{=} 1+\poly{\pi_1}+\poly{\pi_2}
            \overset{IH}{>} 1+\poly{\rho'_k}+\poly{\pi_2}
            \overset{\text{Eq.~\ref{eq:dec_let_2_2}}}{=} \poly{\rho_k}
        \end{align*}
        this holds for every $k$ and so we can conclude.\\

        \item \textit{Case $\pi\triangleright  \vdashc[\Theta] \letin{\lis{x_1,x_2}=\lis{Z_1,Z_2}}{N}:A$ and $\smstep[\Theta][A][n](M,e)=\delta(N[Z_1/x_1,Z_2/x_2],e)$}:\\
        By LET\_PAIR-typing rule we have $\pi_1 \triangleright  \vdashv[\Theta] \lis{Z_1,Z_2}:P_1 \tensor P_2$ and $\pi_2 \triangleright x_1:P_1,x_2:P_2 \vdashc[\Theta] N:A$ and by definition 
        \begin{equation}
            \label{eq:dec_let_pair_1}
            \poly{\pi}=1+\poly{\pi_1}+\poly{\pi_2}.
        \end{equation}
        Moreover, by TENSOR-typing rule applied to $\pi_1$ we have  $\pi'_1 \triangleright  \vdashv[\Theta] Z_1:P_1$ and $\pi''_1 \triangleright  \vdashv[\Theta] Z_2:P_2$, and by definition
        \begin{equation}
            \label{eq:dec_let_pair_2}
            \poly{\pi_1}=1+\poly{\pi'_1}+\poly{\pi''_1}.
        \end{equation}
        By Lemma~\ref{lemma:substitution_lamBLL} applied twice, we have that there exists $\pi_2[\pi'_1/x_1,\pi''_1/x_2] \triangleright \vdashc[\Theta] N[Z_1/x_1,Z_2/x_2]:A$ such that 
        $\poly{\pi_2[\pi'_1/x_1,\pi''_1/x_2]}\leq \poly{\pi_2}+\polyStar{\pi'_1}+\polyStar{\pi''_1}$.
        Moreover, we have that $\polyStar{\cdot}$ is always lower or equal to $\poly{\cdot}$, then we have:
        \begin{equation}
            \label{eq:dec_let_pair_3}
            \poly{\pi_2[\pi'_1/x_1,\pi''_1/x_2]}\leq \poly{\pi_2}+\poly{\pi'_1}+\poly{\pi''_1}.
        \end{equation} 
        Therefore, we have that
        \begin{align*}
            \poly{\pi} &\overset{\text{Eq.~\ref{eq:dec_let_pair_1}}}{=} 1+\poly{\pi_1}+\poly{\pi_2}\\ 
            &\hspace{2mm} > \poly{\pi_2[\pi'_1/x_1,\pi''_1/x_2]}\\
            &\overset{\text{Eq.~\ref{eq:dec_let_pair_3}}}{\leq} 1+\poly{\pi_2}+\poly{\pi'_1}+\poly{\pi''_1}\\ 
            &\overset{\text{Eq.~\ref{eq:dec_let_pair_2}}}{=} 1+ \poly{\pi_1}+\poly{\pi_2}
        \end{align*}
        and we can conclude.\\

        \item \textit{Case $\pi\triangleright  \vdashc[\Theta] \der(!N):A$ and $\smstep[\Theta][A][n](M,e)=\delta(N,e)$}:\\
        By DER-typing rule we have $\pi_1 \triangleright  \vdashv[\Theta] !N:!_1 A$ and by definition
        \begin{equation}
            \label{eq:dec_der_1}
            \poly{\pi}=1+\poly{\pi_1}.
        \end{equation} 
        Moreover, by BANG-typing rule applied to $\pi_1$ we have $\pi'_1 \triangleright  \vdashc[\Theta] N:A$ and by definition
        \begin{equation}
            \label{eq:dec_der_2}
            \poly{\pi_1}=1\times\poly{\pi'_1}=\poly{\pi'_1}.
        \end{equation} 
        Therefore, we have that
        \begin{align*}
            \poly{\pi} &\overset{\text{Eq.~\ref{eq:dec_der_1}}}{=} 1+\poly{\pi_1}  
            > \poly{\pi'_1}
            \overset{\text{Eq.~\ref{eq:dec_der_2}}}{=} \poly{\pi_1}
        \end{align*}
        and we can conclude.\\

        \item \textit{Case $\pi\triangleright  \vdashc[\Theta] (\return{(\lambda x.N)})Z:A$ and $\smstep[\Theta][A][n](M,e)=\delta(N[Z/x],e)$}:\\
        Since by hypothesis we have $\pi\triangleright  \vdashc[\Theta] (\return{(\lambda x.N)})Z:A$, we proceed by analyzing such derivation as follows:
        \begin{itemize}
            \item By APP-typing rule we have $\pi_1 \triangleright  \vdashc[\Theta] \return{(\lambda x.N)}:P\multimap A$ and $\pi_2 \triangleright  \vdashv[\Theta] Z:P$, and by definition we have $\poly{\pi}=1+\poly{\pi_1}+\poly{\pi_2}$.
            \item By ETA-typing rule applied to $\pi_1$ we have $\pi'_1 \triangleright  \vdashv[\Theta] \lambda x.N:P\multimap A$ and by definition $\poly{\pi_1}=\poly{\pi'_1}$.
            \item By LAM-typing rule applied to $\pi'_1$ we have $\pi''_1 \triangleright x:P \vdashc[\Theta] N:A$ and by definition $\poly{\pi'_1}=1+\poly{\pi''_1}$.
        \end{itemize}
        Summing up, we gave that 
        \begin{equation}
            \label{eq:dec_app_1_1}
            \poly{\pi}=2+\poly{\pi''_1}+\poly{\pi_2}.
        \end{equation}
        By Lemma~\ref{lemma:substitution_lamBLL} applied to $\pi''_1$ and $\pi_2$, we have that there exists $\pi''_1[\pi_2/x] \triangleright \vdashc[\Theta] N[Z/x]:A$ such that 
        $\poly{\pi''_1[\pi_2/x] }\leq \poly{\pi''_1}+\polyStar{\pi_2}$.
        Moreover, we have that $\polyStar{\cdot}$ is always lower or equal to $\poly{\cdot}$, then we have:
        \begin{equation}
            \label{eq:dec_app_1_2}
            \poly{\pi''_1[\pi_2/x] }\leq \poly{\pi''_1}+\poly{\pi_2}.
        \end{equation}  
        Therefore, we have that
        $$
        \poly{\pi}\overset{\text{Eq.~\ref{eq:dec_app_1_1}}}{=} 2+\poly{\pi''_1}+\poly{\pi_2}
        > \poly{\pi''_1[\pi_2/x] }
        \overset{\text{Eq.~\ref{eq:dec_app_1_2}}}{\leq} \poly{\pi''_1}+\poly{\pi_2}
        $$
        and we can conclude.\\

        \item \textit{Case $\pi\triangleright  \vdashc[\Theta] \ifthen{\true}{N_1}{N_2}:A$ and $\smstep[\Theta][A][n](M,e)=\delta(N_1,e)$}:\\
        By CASE-typing rule we have $\pi_g \triangleright  \vdashv[\Theta] \true:\BB$, $\pi_1 \triangleright  \vdashc[\Theta] N_1:A$ and $\pi_2 \triangleright  \vdashc[\Theta] N_2:A$ and by definition 
        \begin{equation}
            \label{eq:dec_if_true}
            \poly{\pi}=1+\poly{\pi_g}+ \poly{\pi_1}+\poly{\pi_2}.
        \end{equation} 
        Therefore, we have that
        $$
        \poly{\pi}\overset{\text{Eq.~\ref{eq:dec_if_true}}}{=} 1+\poly{\pi_g}+\poly{\pi_1}+\poly{\pi_2}
        > \poly{\pi_1}
        $$
        and we can conclude.\\

        \item \textit{Case $\pi\triangleright  \vdashc[\Theta] \ifthen{\false}{N_1}{N_2}:A$ and $\smstep[\Theta][A][n](M,e)=\delta(N_2,e)$}:\\
        By CASE-typing rule we have $\pi_g \triangleright  \vdashv[\Theta] \true:\BB$, $\pi_1 \triangleright  \vdashc[\Theta] N_1:A$ and $\pi_2 \triangleright  \vdashc[\Theta] N_2:A$ and by definition 
        \begin{equation}
            \label{eq:dec_if_false}
            \poly{\pi}=1+\poly{\pi_g}+\poly{\pi_1}+\poly{\pi_2}.
        \end{equation} 
        Therefore, we have that
        $$
        \poly{\pi}\overset{\text{Eq.~\ref{eq:dec_if_false}}}{=} 1+\poly{\pi_g}+\poly{\pi_1}+\poly{\pi_2}
        > \poly{\pi_2}
        $$
        and we can conclude.\\

        \item \textit{Case $\pi\triangleright  \vdashc[\Theta] \setr{r}{V}:\UU$ and $\smstep[\Theta][A][n](M,e)=\delta(\return{\star},e[V/r])$}:\\
        By SET-typing rule we have $\pi_1 \triangleright  \vdashv[\Theta] V:G$ and $r:G\in\Theta$, by definition
        \begin{equation}
            \label{eq:dec_set_1}
            \poly{\pi}=2+\poly{\pi_1}.
        \end{equation}\\
        Let $\rho \triangleright  \vdashc[\Theta] \return{\star}: \UU$, then by ETA-typing rule we have  $\rho' \triangleright  \vdashv[\Theta] \star: \UU$ and by definition
        \begin{equation}
            \label{eq:dec_set_2}
            \poly{\rho}=1+\poly{\rho'}=1+1=2.
        \end{equation}
        Therefore, we have that
        $
        \poly{\pi} \overset{\text{Eq.~\ref{eq:dec_set_1}}}{=} 2+\poly{\pi_1} 
        > \poly{\rho} 
        \overset{\text{Eq.~\ref{eq:dec_set_2}}}{=} 2
        $
        and we can conclude since it is easy to check that $\forall n\geq 1.\mbox{ } \poly{\pi_1}n\geq 1$.\\

        \item \textit{Case $\pi\triangleright  \vdashc[\Theta] \get{r}:G$ and $\smstep[\Theta][A][n](M,e)=\delta(\return{e(r)},e)$}:\\
        By GET-typing rule we have $r:G\in\Theta$ and by definition
        \begin{equation}
            \label{eq:dec_get_1}
            \poly{\pi}=2+\sizeG{G}.
        \end{equation}
        Let $\rho \triangleright  \vdashc[\Theta] \return{e(r)}: G$, then by ETA-typing rule and since $e(r)$ is of ground type $G$ and by Lemma~\ref{lemma:typableValuesGroudType} we have  $\rho' \triangleright  \vdashv[\varnothing] e(r): G$. By definition $\poly{\rho}=1+\poly{\rho'}$. Moreover, by Lemma~\ref{lemma:costValueGroundType} applied to $\rho'$ we have that $\poly{\rho'}=\sizeG{G}$. Summing up, we have 
        \begin{equation}
            \label{eq:dec_get_2}
            \poly{\rho}=1+\sizeG{G}.
        \end{equation}
        Therefore, we have that
        $
        \poly{\pi} \overset{\text{Eq.~\ref{eq:dec_get_1}}}{=} 2+\sizeG{G} 
        > \poly{\rho} 
        \overset{\text{Eq.~\ref{eq:dec_get_2}}}{=} 1+\sizeG{G}
        $
        and we can conclude.\\

        \item \textit{Case $\pi\triangleright  \vdashc[\Theta] f_p(Z_1,\ldots,Z_m):Gp$ and $\smstep[\Theta][A][n](M,e)=\sum_k a_k \mbox{ } \delta(U_k,e_k)=\den{f}_p(n) (Z_1,\ldots,Z_m)$}:\\
        By FUN-typing rule we have $\typeof{f}=G_1 \times \ldots \times G_m \rightarrow G$ and $\pi_{1j}\triangleright  \vdashv[\Theta] Z_j:G_j$ for $1\leq j\leq m$. Moreover, by definition we have 
        \begin{equation}
            \label{eq:dec_fun_1}
            \poly{\pi}=1+\sizeG{G}+\sum_{j=1}^{m} \poly{\pi_{1j}}.
        \end{equation} 
        Since $U_k \in \den{Gp}_n$ for every $k$, by Lemma~\ref{lemma:typableValuesGroudType} we have $\rho_k \triangleright\vdashv[\varnothing] U_k:Gp(n)$. Moreover, by Lemma~\ref{lemma:costValueGroundType} applied to $\rho_k$ for every $k$, we have that
        \begin{equation}
            \label{eq:dec_fun_3}
            \poly{\rho_k}=\sizeG{G}.
        \end{equation}
        Therefore, we have that
        $$
        \poly{\pi} \overset{\text{Eq.~\ref{eq:dec_fun_1}}}{=} 1+\sizeG{G}+\sum_{j=1}^{m} \poly{\pi_{1j}} 
        > \poly{\rho_k} 
        \overset{\text{Eq.~\ref{eq:dec_fun_3}}}{=} \sizeG{G}
        $$
        this holds for every $k$ and so we can conclude.\\

        \item \textit{Case $\pi\triangleright  \vdashc[\Theta] \loopp{(\lambda x.N_2)}{1}{N_1}:P$ and $\smstep[\Theta][A][n](M,e)=\sum_k a_k \mbox{ } \delta(\letin{x=N_1}{N_2},e)$}:\\
        Since by hypothesis we have $\pi\triangleright  \vdashc[\Theta] \loopp{(\lambda x.N_2)}{1}{N_1}:P$, we proceed by analyzing such derivation as follows:
        \begin{itemize}
            \item By LOOP-typing rule we have $\pi_1\triangleright  \vdashv[\Theta] \lambda x.N_2:P\multimap P$ and $\pi_2\triangleright  \vdashc[\Theta] N_1:P$. Moreover, by definition
            $\poly{\pi}=1+1+1\times \poly{\pi_1}+\poly{\pi_2}=2+\poly{\pi_1}+\poly{\pi_2}$.
            \item By LAM-typing rule applied to $\pi_1$ we have $\pi'_1\triangleright  x:P \vdashc[\Theta] N_2:P$. Moreover, by definition we have that $\poly{\pi_1}=1+\poly{\pi'_1}$.
        \end{itemize}
        Summing up, we have that
        \begin{equation}
            \label{eq:dec_loop_1_1}
            \poly{\pi}=3+\poly{\pi'_1}+\poly{\pi_2}.
        \end{equation}
        Let $\rho \triangleright  \vdashc[\Theta] \letin{x=N_1}{N_2}:P$, then by definition 
        \begin{equation}
            \label{eq:dec_loop_1_2}
            \poly{\rho}=1+\poly{\pi'_1}+\poly{\pi_2}.
        \end{equation} 
        Therefore, we have that
        $$
        \poly{\pi} \overset{\text{Eq.~\ref{eq:dec_loop_1_1}}}{=} 3+\poly{\pi'_1}+\poly{\pi_2} 
        > \poly{\rho} 
        \overset{\text{Eq.~\ref{eq:dec_loop_1_2}}}{=} 1+\poly{\pi'_1}+\poly{\pi_2}
        $$
        and we can conclude.\\

        \item \textit{Case $\pi\triangleright  \vdashc[\Theta] \loopp{(\lambda x.N_2)}{k+1}{N_1}:P$ and \\$\smstep[\Theta][A][n](M,e)=\sum_k a_k \mbox{ } \delta(\letin{x=(\loopp{(\lambda x.N_2)}{k}{N_1})}{N_2},e)$}:\\
        Since by hypothesis we have $\pi\triangleright  \vdashc[\Theta] \loopp{(\lambda x.N_2)}{k+1}{N_1}:P$, we proceed by analyzing such derivation as follows:
        \begin{itemize}
            \item By LOOP-typing rule we have $\pi_1\triangleright  \vdashv[\Theta] \lambda x.N_2:P\multimap P$ and $\pi_2\triangleright  \vdashc[\Theta] N_1:P$. Moreover, by definition
            $\poly{\pi}=1+k+1+(k+1)\times \poly{\pi_1}+\poly{\pi_2}=2+k+(k+1)\times \poly{\pi_1}+\poly{\pi_2}$.
            \item By LAM-typing rule applied to $\pi_1$ we have $\pi'_1\triangleright  x:P \vdashc[\Theta] N_2:P$. Moreover, by definition we have that $\poly{\pi_1}=1+\poly{\pi'_1}$.
        \end{itemize}
        Summing up, we have that
        \begin{align*}
            \poly{\pi}&=2+k+(k+1)\times \poly{\pi_1}+\poly{\pi_2}\\
            &=2+k+(k+1)\times(1+\poly{\pi'_1})+\poly{\pi_2}\\
            &=2+k+(k+1)\times 1+(k+1)\times\poly{\pi'_1}+\poly{\pi_2}\\
            &=2+k+k+1+(k+1)\times\poly{\pi'_1}+\poly{\pi_2}\\
            &=3+2k+(k+1)\times\poly{\pi'_1}+\poly{\pi_2}
        \end{align*}
        Thus, we have
        \begin{equation}
            \label{eq:dec_loop_2_1}
            \poly{\pi}=3+2k+(k+1)\times\poly{\pi'_1}+\poly{\pi_2}.
        \end{equation}
        Let $\rho \triangleright  \vdashc[\Theta] \letin{x=(\loopp{(\lambda x.N_2)}{k}{N_1})}{N_2}:P$ , we proceed by analyzing such derivation as follows:
        \begin{itemize}
            \item By LET-typing rule we have $\rho_1 \vdashc[\Theta] \loopp{(\lambda x.N_2)}{k}{N_1}:P$ and $\pi'_1\triangleright  x:P \vdashc[\Theta] N_2:P$. By definition we have $\poly{\rho}=1+\poly{\rho_1}+\poly{\pi'_1}$.
            \item By LOOP-typing rule applied to $\rho_1$ we have $\pi_1\triangleright  \vdashv[\Theta] \lambda x.N_2:P\multimap P$ and $\pi_2\triangleright  \vdashc[\Theta] N_1:P$. Moreover, by definition
            $\poly{\rho_1}=1+k+k\times \poly{\pi_1}+\poly{\pi_2}$.
            \item By LAM-typing rule applied to $\pi_1$ we have $\pi'_1\triangleright  x:P \vdashc[\Theta] N_2:P$. Moreover, by definition we have that $\poly{\pi_1}=1+\poly{\pi'_1}$.
        \end{itemize}
        Summing up, we have that
        \begin{align*}
            \poly{\rho}&=1+\poly{\rho_1}+\poly{\pi'_1}\\
            &=1+\left( 1+k+k\times \poly{\pi_1}+\poly{\pi_2} \right)+\poly{\pi'_1}\\
            &=1+\left( 1+k+k\times (1+\poly{\pi'_1})+\poly{\pi_2} \right)+\poly{\pi'_1}\\
            &=2+k+k\times(1+\poly{\pi'_1})+\poly{\pi_2}+\poly{\pi'_1}\\
            &=2+k+k+k\times\poly{\pi'_1}+\poly{\pi_2}+\poly{\pi'_1}\\
            &=2+2k+k\times\poly{\pi'_1}+\poly{\pi'_1}+\poly{\pi_2}\\
            &=2+2k+(k+1)\times\poly{\pi'_1}+\poly{\pi_2}
        \end{align*}
        Thus, we have
        \begin{equation}
            \label{eq:dec_loop_2_2}
            \poly{\rho}=2+2k+(k+1)\times \poly{\pi'_1}+\poly{\pi_2}
        \end{equation}
        Therefore, we have that
        \begin{align*}
        \poly{\pi} &\overset{\text{Eq.~\ref{eq:dec_loop_2_1}}}{=} 3+2k+(k+1)\times\poly{\pi'_1}+\poly{\pi_2}\\
        &> \poly{\rho} 
        \overset{\text{Eq.~\ref{eq:dec_loop_2_2}}}{=} 2+2k+(k+1)\times \poly{\pi'_1}+\poly{\pi_2}
        \end{align*}
        and we can conclude.\\
    \end{itemize}
\end{proof}

Moreover, we prove the Subject Reduction Theorem (Theorem~\ref{th:SR}), which guarantees two critical properties for well-typed terms in $\lamBLL$. First, it ensures that every term obtained after a single small-step reduction remains well-typed with the same type, preserving type correctness. Second, the theorem guarantees that the assigned cost of the original term is always strictly higher than the costs of any terms it reduces to, reinforcing the idea that evaluation progresses by reducing complexity. Together, these properties are essential to maintain both type safety and the soundness of the polynomial cost analysis during execution.

\begin{theorem}[Quantitative Subject Reduction]\label{th:SR}
    For all $n\geq 1$, for every reference context $\Theta$, type $A$ and term $M\in\Lambda_n(\Theta,A)$, for all $e\in\Store_{\Theta}$ such that $\smstep[\Theta][A][n](M,e)=\sum_{k} a_k \mbox{ } \delta(M_k, e_k)$, we have that:
    \begin{enumerate}
        \item $\forall k. \vdashc[\Theta n] M_k:An$ is derivable;
        \item $\cost_n^{A,\Theta}(M,e)=1+\max_k \cost_n^{A,\Theta}(M_k,e_k)$
    \end{enumerate}
\end{theorem} 
\begin{proof}\leavevmode \vspace{-1mm}
    \begin{enumerate}
        \item The existence of a derivation $\vdashc[\Theta n] M_k:An$ has been already proved in Lemma~\ref{lemma:descreasing}.
        \item By Definition~\ref{def:cost} we have that $\cost_n^{A,\Theta}(M,e)=\inf\limits_k \{ \smstep[\Theta][A][n,k](M,e) \neq \bot \}$, so we proceed by 
        distinguishing two cases:
        \begin{varitemize}
            \item \textit{For every $l\in\mathbb{N}$ it holds that $\smstep[\Theta][A][n,l](M,e)= \bot$}.
          	In this case we have that $\cost_n^{A,\Theta}(M,e)=\infty$, and moreover there must be a $k$ such that
          	$\cost_n^{A,\Theta}(M_k,e_k)=\infty$, because if all those expressions $\cost_n^{A,\Theta}(M_k,e_k)$ were finite, there would be a natural number $m$
          	with $\smstep[\Theta][A][n,l](M,e)\neq\bot$. The equality thus holds.\\
            \item \textit{There is $l\in\mathbb{N}$ such that $\smstep[\Theta][A][n,l](M,e)\neq \bot$}. Suppose that $l$ is the least such $l$.
            We have that $\smstep[\Theta][A][n](M,e)=\sum_{k} a_k \mbox{ } \delta(M_k, e_k)$. 
            Moreover, we have that $\smstep[\Theta][A][n,l](M,e)=\sum_k a_k\mbox{ }\Dist_k$ where $\Dist_k=\smstep[\Theta][A][n,l_k](M_k,e_k)$ 
            and $l=1+\max_k l_k$. Therefore, we have that 
            $$
            \cost_n^{A,\Theta}(M,e)=\inf\limits_k \{ \smstep[\Theta][A][n,k](M,e) \neq \bot \}=l=1+\max_k l_k\geq 1+\max_k \cost_n^{A,\Theta}(M_k,e_k)
            $$
            Suppose now, by way of contradiction, that the last inequality is strict. That would mean that all the $l_k$'s having the highest value
            would be greater than  $\max_k \cost_n^{A,\Theta}(M_k,e_k)$, but this would go against the minimality of $l$, because they could all be decreased safely
            without incurring in $\bot$. We can thus conclude.
        \end{varitemize}
    \end{enumerate}
\end{proof}

Finally, Theorem~\ref{th:polysound} formalizes the key result that all well-typed term in $\lamBLL$ evaluates to normal form within a number of steps bounded by a polynomial in the size of its input. The proof proceeds by leveraging the strictly decreasing nature of the polynomial cost assignment along reduction steps (Lemma~\ref{lemma:descreasing}) and the preservation of typing and cost bounds under reduction (Theorem~\ref{th:SR}). 
 
\PolySoundness*
\begin{proof}
    By Theorem~\ref{th:progress} we have two cases:
    \begin{enumerate}
        \item \textit{Case $M=\return{V}$}:\\
        We have that $\cost_n^{A,\Theta}(M,e)=0$ and so the inequality is trivial.
        \item \textit{Case $\smstep[\Theta][A][n](M,e)=\sum_k a_k\mbox{ }\delta(M_k,e_k)$}:\\
        By Lemma~\ref{lemma:descreasing} we have that for every $k$ there exists a derivation $\rho_k \triangleright  \vdashc[\Theta] M_k:A$ such that $\poly{\pi}>\poly{\rho_k}$.
        Therefore, we have that 
        \begin{equation}
            \label{eq:polysound_1}
            1+\max_k\poly{\rho_k}(n) \leq \poly{\pi}(n).
        \end{equation}
        By Theorem~\ref{th:SR} case 2, we have that
        \begin{equation}
            \label{eq:polysound_2}
            \cost_n^{A,\Theta}(M,e)=1+\max_k \mbox{ }\cost_n^{A,\Theta}(M_k,e_k).
        \end{equation}
        Summing up, we have
        \begin{align*}
            \cost_n^{A,\Theta}(M,e) 
            &\overset{\text{Eq.}~\ref{eq:polysound_2}}{=} 1+\max_k\mbox{ }\cost_n^{A,\Theta}(M_k,e_k)\\
            &\hspace{2mm}\leq 1+\max_k\mbox{ }\poly{\rho_k}(n)\\
            &\overset{\text{Eq.}~\ref{eq:polysound_1}}{\leq}
            \poly{\pi}(n)
        \end{align*}
        and we can conclude.
    \end{enumerate}
\end{proof}
\section{Completeness}
\label{app:completeness} 
In Subsection~\ref{subsect:completeness}, we have described the PTM of interest, our goal is to establish the completeness theorem, namely to prove that, given a first-order interface $I=(p,q)$, every $(p,q)$-function $f$ computed by a PTM $\mathcal{M}$ working in polynomial time can be represented by a $\lamBLL$ term which computes it. To achieve this, we first present the encoding of the PTM within 
$\lamBLL$ and then outline the proof.

\subsection{Encoding of a PTM in $\lamBLL$}
First, we systematically encode the machine's components and its computation within $\lamBLL$ by using a minimal set of basic functions symbols. 

\subsubsection{Encoding of Configurations.}
In order to model the tapes and their content, we first need to encode them in a binary-friendly format. The tape alphabet of a PTM is $\tapeAlph=\{0,1,\blank\}$, whereas $\lamBLL$ operates exclusively on binary strings. To bridge this, we assign a fixed two-bit code to each tape symbol: $0$ is encoded as $00$, $1$ as $01$, and $\blank$ as $10$. In this way, the entire tape can be represented as a binary string composed of two-bit codes. Thus, the type $\typSymb$ is defined as just $\typSymb=\Str[1]\otimes\Str[1]$.

Let $r\in\N_{\geq 1}[i]$ be the runtime of the PTM we want to encode. Each tape in the configuration receives the following type in $\lamBLL$:
$$\typTape=\Str[2r] \tensor \typSymb \tensor \Str[2r]$$
where the first component corresponds to the portion of the tape to the left of the tape head, the second component encodes the symbol currently under the tape head, and the third component corresponds to the portion of the tape to the left of the tape head. Let $h=|Q|$ denote the number of states of the PTM, we encode the current state by a binary string of type $\typState=\Str[\lceil log_2\mbox{ }h \rceil]$, sufficient to represent all $h$ possible states.

Combining these components, a full configuration $C$ consisting of the current state and the contents of three tapes, can be represented in $\lamBLL$ by a value
$\enc{C}$ of the type
$$ \typConf=\typState \tensor \typTape \tensor \typTape \tensor \typTape .$$
 
\subsubsection{Encoding of Transition Functions.}
First, we encode the set of possible directions that the tape head can move. Let $\{L,R,S\}$ be the set of tape head directions. Analogously to the encoding of the tape alphabet, each tape head direction is assigned a fixed binary code. Specifically, moving left $L$ is encoded as $00$, moving right $R$ as $01$, and staying in place $S$ as $10$. 
This way, we get the type $\typDir=\Str[1]\otimes\Str[1]$.
With this encoding, we can represents the transition functions $\delta_0$ and $\delta_1$ of the PTM as two terms $\transitionTerm{0}$ and $\transitionTerm{1}$ in $\lamBLL$, typed as follows:
\begin{align*}
    \vdashv[\varnothing] \transitionTerm{b}:
    \left(
        \begin{aligned}
            &\typState\\
            \tensor & \typSymb\\
            \tensor & \typSymb
        \end{aligned}
    \right)
    \lollipop
     \left(
        \begin{aligned}
            &\typState\\
            \tensor & \typDir\\
            \tensor & \typSymb \tensor \typDir\\
             \tensor& \typSymb
        \end{aligned}
    \right)
\end{align*}
where $b\in\{0,1\}$. The two symbols in input come from the input and work tape, while in output we produce a direction for the input tape, a symbol and direction for the work tape, and a symbol for the output (where the blank symbol stands for the intention of not writing anything, and keep the head where it is).
A transition function $\delta_b$  with $b\in\{0,1\}$ can be seen as a cases inspection over the input tuple which returns a branch-specific action tuple. The term $\transitionTerm{b}$ with $b\in\{0,1\}$ simply implements this case inspection by nested if-then-else.

\subsubsection{Encoding Actions on Tapes.}
Given a tape represented as a tuple, our goal is to encode the operation of a PTM that writes a new symbol on the tape and moves the tape head either to the left or to the right. We want to express this entire process as $\lamBLL$ terms. 

First, we introduce three auxiliary function symbols that operate on binary strings. These symbols are defined with the following types:
\begin{align*}
    \typeof{\texttt{cyclicShiftL}}&: \Str[i] \rightarrow \Str[i]\\
    \typeof{\texttt{cyclicShiftR}}&: \Str[i] \rightarrow \Str[i]\\
    \typeof{\texttt{write}}&:\Str[i]\times\Str[1] \rightarrow \Str[i]\\
    \typeof{\texttt{read}}&:\Str[i] \rightarrow \Str[1]
\end{align*}
The function symbol $\texttt{cyclicShiftL}$ performs a left cyclic shift, moving each bit one position to the left and sending the first bit to the end of the string. Conversely, $\texttt{cyclicShiftR}$ performs a right cyclic shift, moving each bit one position to the right and placing the last bit at the beginning. For example, $\texttt{cyclicShiftL}(1100)=1001$ and $\texttt{cyclicShiftR}(1100)=0110$. The symbol $\texttt{write}$ models the act of writing on a string by replacing the first bit of a binary string with a chosen value. Complementarily, the function symbol $\texttt{read}$ encodes the reading operation, returning the first bit of a binary string. Together, these four primitives provide the essential building blocks for encoding the PTM's tape update mechanism as $\lamBLL$ terms, and are sufficient for our purposes.

The $\lamBLL$ term encoding the writing action on a tape is typed as
$$\vdashv[\varnothing] \writeSymbTerm : \typTape \tensor\typSymb\lollipop \typTape$$
indicating that it takes as input a pair consisting of the current tape configuration and a binary-encoded symbol, and produces a new tape with the symbol written at the head position. 
Its implementation follows the intuitive idea of destructuring the tape into its components, updating the current symbol using the auxiliary $\texttt{write}$ function on the binary encoding, and reconstructing the tape. Formally, it can be defined as follows:
\begin{align*}
    \writeSymbTerm \triangleq \mbox{ }&
    \lambda e.\\
    &\letin{\lis{t,\mathit{newsymb}}=e}{}\\
    &\letin{\lis{\mathit{leftTape},\mathit{symb},\mathit{rightTape}}=t}{}\\
    &\return{\lis{\mathit{leftTape}, \mathit{newsymb}, \mathit{rightTape}}}
\end{align*}
The $\lamBLL$ terms encoding the movement left and right of a tape head are typed as
\begin{align*}
    &\vdashv[\varnothing] \moveLTerm : \typTape \lollipop \typTape\\
    &\vdashv[\varnothing] \moveRTerm : \typTape \lollipop \typTape
\end{align*}
indicating that each term takes a tape as input and produces a new tape with the head moved one position to the left or right, respectively.
Formally, the definitions, shown in Figure~\ref{fig:moveHead}, implement these movements by applying the auxiliary cyclic shift and write/read operations to the left and right tape segments. The terms reorganize the tape's binary representation to reflect the head's new position while preserving the surrounding symbols and updating the relevant bits under the head.
\begin{figure}[h]
    \centering
    \makebox[\textwidth][c]{%
    \begin{minipage}{1.2\textwidth}
    \begin{subfigure}{0.45\textwidth}
        \begin{align*}
            \textit{MoveL}&\textit{eft} \mbox{ }\triangleq\\
            &\lambda t.\\
            &\letin{\lis{\mathit{leftTape},\mathit{symb},\mathit{rightTape}}=t}{}\\
            &\letin{\lis{s_1,s_2}=\mathit{symb}}{}\\
            &\letin{\mathit{tr}_1=\texttt{cyclicShiftR}(\mathit{rightTape})}{}\\
            &\letin{\mathit{tr}_2=\texttt{write}(\mathit{tr}_1,s_2)}{}\\
            &\letin{\mathit{tr}_3=\texttt{cyclicShiftR}(\mathit{tr}_2)}{}\\
            &\letin{\mathit{tr}_4=\texttt{write}(\mathit{tr}_3,s_1)}{}\\
            &\letin{g_2=\texttt{read}(\mathit{leftTape})}{}\\
            &\letin{\mathit{tl}_1=\texttt{write}(\mathit{leftTape},0)}{}\\
            &\letin{\mathit{tl}_2=\texttt{cyclicShiftR}(\mathit{tl}_1)}{}\\
            &\letin{g_1=\texttt{read}(\mathit{leftTape})}{}\\
            &\letin{\mathit{tl}_3=\texttt{write}(\mathit{tl}_2,1)}{}\\
            &\letin{\mathit{tl}_4=\texttt{cyclicShiftR}(\mathit{tl}_3)}{}\\
            &\return{\lis{\mathit{tl}_4,\lis{g_1,g_2},\mathit{tr}_4}}
        \end{align*}
    \caption{Term for moving the tape head to the left.}
    \label{subfig:moveL}
    \end{subfigure}
    \quad 
    \begin{subfigure}{0.45\textwidth}
        \begin{align*}
            \textit{MoveR}&\textit{ight} \mbox{ }\triangleq\\ 
            &\lambda t.\\
            &\letin{\lis{\mathit{leftTape},\mathit{symb},\mathit{rightTape}}=t}{}\\
            &\letin{\lis{s_1,s_2}=\mathit{symb}}{}\\
            &\letin{\mathit{tl}_1 = \texttt{cyclicShiftL}(\mathit{leftTape})}{}\\
            &\letin{\mathit{tl}_2 = \texttt{write}(\mathit{tl}_1,s_1)}{}\\
            &\letin{\mathit{tl}_3 = \texttt{cyclicShiftL}(\mathit{tl}_2)}{}\\
            &\letin{\mathit{tl}_4 = \texttt{write}(\mathit{tl}_3,s_2)}{}\\
            &\letin{g_1=\texttt{read}(\mathit{rightTape})}{}\\
            &\letin{\mathit{tr}_1 = \texttt{write}(\mathit{rightTape}, 1)}{}\\
            &\letin{\mathit{tr}_2 = \texttt{cyclicShiftL}(\mathit{tr}_1)}{}\\
            &\letin{g_2=\texttt{read}(\mathit{tr}_2)}{}\\
            &\letin{\mathit{tr}_3 = \texttt{write}(\mathit{tr}_3, 0)}{}\\
            &\letin{\mathit{tr}_4 = \texttt{cyclicShiftL}(\mathit{tr}_2)}{}\\
            &\return{\lis{\mathit{tl}_4,\lis{g_1,g_2},\mathit{tr}_4}}
        \end{align*}
    \caption{Term for moving the tape head to the right.}
    \label{subfig:moveR}
    \end{subfigure}
   
    \caption{Implementation of $\lamBLL$ terms for tape head movement. Each subfigure shows the step-by-step construction of the new tape configuration when moving left or right.}
    \label{fig:moveHead}
     \end{minipage}%
    }
\end{figure}

\subsubsection{Encoding a Step of Computation of the PTM.}
In order to encode a single step of computation of the PTM as a term in $\lamBLL$, we introduce a function symbol $\random$ typed as
\begin{align*}
    \typeof{\random}&: \Str[i]
\end{align*} 
which returns a randomly selected string. Out of it, one can define easily a term $\flipcoin$ having type
$$\vdashc[\varnothing] \flipcoin : \BB$$
and simulating the outcome of a fair coin toss. Moreover, one can easily define a term
$\equals$ which test whether two single-bit strings are equalt, and having type
$$
\vdashv[\varnothing] \equals : \Str[1]\tensor\Str[1]\lollipop\BB
$$ 

We begin by defining the $\lamBLL$ term $\moveHeadTerm$, which takes a tape and an encoded action to be performed on it, and updates both the tape's content and the head position accordingly, using previously defined terms.
Formally, $\moveHeadTerm$ is typed as:
$$\vdashv[\varnothing] \moveHeadTerm : \typTape \tensor \typDir\lollipop \typTape$$
The term in Figure~\ref{fig:moveHeads} inspects the encoding of the direction: if it represents left, it applies $\moveLTerm$; if right,  $\moveRTerm$; otherwise, the head remains in place. This modular design allows $\moveHeadTerm$ to handle any of the three tapes uniformly, without duplication of logic.

\begin{figure}
    \begin{align*}
        \moveHeadTerm \triangleq \mbox{ }&
        \lambda e.\\
        &\letin{\lis{t,d}=e}{}\\
        &\letin{\lis{d_1,d_2}=d}{}\\
        &\letin{b_1 = \texttt{equal}(d_1,0)}{} \\
        &\texttt{if } b_1 \texttt{ then}
        &&\text{//d is the encoding for the direction L or R}\\
        &\quad \letin{b_2 = \texttt{equal}(d_2,0)}{}\\
        &\quad \texttt{if } b_2 \texttt{ then } \return \moveLTerm \mbox{ } t 
        \qquad &&\text{//d is the encoding for the direction L }\\
        &\quad \texttt{else  } \return \moveRTerm \mbox{ } t  
        \qquad && \text{//d is the encoding for the direction R }\\
        &\texttt{else }  \return t
        &&\text{//d is the encoding for the direction S}
    \end{align*}
    \caption{Implementation of the $\lamBLL$ term that updates a tape based on the encoded head movement (left, right, or stay).}
    \label{fig:moveHeads}
\end{figure}

Using $\moveHeadTerm$ as subterm, we define the term $\evalTerm$, which executes a single transition of the PTM by applying a transition description to the current configuration. Its type is:
\begin{align*}
    \vdashv[\varnothing] \evalTerm :
    \left(\left(
        \begin{aligned}
            &\typState\\
            \tensor & \typDir\\
            \tensor & \typSymb\tensor\typDir\\
            \tensor & \typSymb
        \end{aligned}
    \right)
    \tensor \typConf \right) \lollipop \typConf
\end{align*} 
Given a transition and the current configuration, $\evalTerm$ constructs the next configuration by first writing the specified symbol on the work tape using $\writeSymbTerm$, then moving the heads of the input and work tapes according to the encoded directions using $\moveHeadTerm$, and finally conditionally writing a symbol on the output tape depending on the encoded output action. Figure~\ref{subfig:eval} shows the implementation of $\evalTerm$, which is carefully structured so that each subterm corresponds to a specific operational aspect of the PTM step.

Finally, the term $\stepPTMTerm$ defines the $\lamBLL$ encoding of a single probabilistic transition of the PTM, combining deterministic state evolution with random choice. Its type, 
$$\vdashv[\varnothing]\stepPTMTerm:\typConf\lollipop\typConf$$
indicates that the term takes a machine configuration as input and produces the next configuration resulting from one computation step. 

Operationally, $\stepPTMTerm$ first flips a coin using $\flipcoin$ to determine which transition function $\delta_b$ to apply. Depending on the coin's outcome, it retrieves the corresponding transition via $\transitionTerm{b}$ and delegates to $\evalTerm$ to compute the next configuration.
This construction provides a compositional encoding of the PTM’s probabilistic dynamics within $\lamBLL$, faithfully reflecting the behavior described by the \textit{STEP} rule in Definition~\ref{def:stepRelationPTM}.
\begin{figure}
    \begin{subfigure}{0.48\textwidth}
        \begin{align*}
        \evalTerm &\triangleq\\
        &\lambda e.\\
        &\letin{\lis{t,c}=e}{}\\
        &\letin{\lis{q',\mathit{symbW}',aI, aW, aO}=t}{}\\
        &\letin{\lis{q,i,w,o}=c}{}\\
        &\letin{\mathit{tapeW_1}=\writeSymbTerm \mbox{ } \lis{w,\mathit{symbW}'}}{}\\
        &\letin{\mathit{tapeI}=\moveHeadTerm \mbox{ } \lis{i,aI}}{}\\
        &\letin{\mathit{tapeW_2}=\moveHeadTerm \mbox{ } \lis{\mathit{tapeW_1},aW}}{}\\
        &\letin{\lis{aO_1,aO_2}=aO}{}\\
        &\letin{\mathit{test}=\texttt{equal}(aO_1,0)}{}\\
        &\texttt{if } \mathit{test} \texttt{ then}\\
        &\quad \letin{\mathit{symbOut}=\lis{0,aO_2}}{}\\
        &\quad \letin{\mathit{tapeO_1}=\writeSymbTerm \mbox{ }\lis{o,\mathit{symbOut}}}{}\\
        &\quad \return{\lis{q', \mathit{tapeI},\mathit{tapeW_2},\mathit{tapeO_1}}}\\
        &\texttt{else}\\
        &\quad \return{\lis{q', \mathit{tapeI},\mathit{tapeW_2},o}}
    \end{align*}
    \caption{Implementation of the $\lamBLL$ term that performs a PTM transition step given a transition description and a current configuration.}
    \label{subfig:eval}
    \end{subfigure}
    \mbox{ }
    \begin{subfigure}{0.48\textwidth}
        \begin{align*}
            \stepPTMTerm & \triangleq \\
            &\lambda c.\\
            &\letin{\lis{q,i,w,o}=c}{}\\
            &\letin{\lis{\mathit{leftI},\mathit{symbI},\mathit{rightI}}=i}{}\\
            &\letin{\lis{\mathit{leftW},\mathit{symbW},\mathit{rightW}}=w}{}\\
            &\letin{\lis{\mathit{leftO},\mathit{symbO},\mathit{rightO}}=o}{}\\
            &\letin{b=\flipcoin}{}\\
            &\texttt{if } b \texttt{ then}\\
            &\quad \letin{t=\transitionTerm{1} \mbox{ }\lis{q,\mathit{symbI}, \mathit{symbW}}}{}\\
            &\quad \evalTerm \mbox{ }\lis{t,c}\\
            &\texttt{else }\\
             &\quad \letin{t=\transitionTerm{0} \mbox{ }\lis{q,\mathit{symbI}, \mathit{symbW}}}{}\\
            &\quad \evalTerm \mbox{ }\lis{t,c}\\
        \end{align*}
        \caption{Implementation of the $\lamBLL$ term encoding a single probabilistic PTM step by sampling a random bit and applying the selected transition.}
        \label{subfig:step}
    \end{subfigure}
    \caption{Encoding of a single PTM computation step in $\lamBLL$: Subfigure~\ref{subfig:eval} depicts the deterministic transition execution, while Subfigure~\ref{subfig:step} incorporates probabilistic choice through $\flipcoin$.}
\end{figure}

\subsubsection{Encoding of Initial Configuration.} 
To represent the initial configuration of a PTM in $\lamBLL$, we begin by describing how the machine bootstrapped before any computation takes place. The PTM starts in its designated initial state, denoted by $q_0$, which identifies the control point from which the execution begins. The input tape contains the binary input string $s$, with the head positioned on its first symbol and the remaining part of the string extending to the right, followed by blanks.
Both the work tape and the output tape are initially empty: they contain only blank symbols. The work tape will serve as temporary storage during computation, while the output tape will be gradually filled with symbols as the machine produces its result.

We define the $\lamBLL$ term $\encodeInitialStateTerm$ to represent the PTM's initial state $q_0$ as a binary string of length $\lceil \log_2 h \rceil$. This is correct since $h$ is constant, so the encoding length is fixed and independent of the input. 
Formally, we have:
\[
\vdashv[\varnothing] \encodeInitialStateTerm: \UU \lollipop \typState
\]
\[
\encodeInitialStateTerm \triangleq 
\lambda \unusedvar. \return{\enc{q_0}}
\]
where $\enc{q}$ is the string encoding any state $q$.

To initialize blank tapes in $\lamBLL$, we define the term $\fillBlankStringTerm$ in Figure~\ref{subfig:fillBlankString}, which generates a string of fixed length $2r$ with alternating zeros and ones. This string serves as the building block for tapes. This term makes use of a nullary function symbol $\texttt{zero}$ returning a string of length $\Str[i]$ filled with zero. 

Using $\fillBlankStringTerm$, the term $\fillBlankTapeTerm$ constructs a full tape by initializing both the left and right halves and pairing them with the initial symbol set (see Figure~\ref{subfig:fillBlankTape}). Together, these terms ensure that the work and output tapes start in a well-structured blank state. Their types are: 
\begin{align*}
    &\vdashv[\varnothing] \fillBlankStringTerm:\UU \lollipop \Str[2r]\\
    &\vdashv[\varnothing] \fillBlankTapeTerm: \UU \lollipop \typTape
\end{align*}
\begin{figure}[h]
    \makebox[\textwidth][c]{%
    \begin{minipage}{1.2\textwidth}

    \begin{subfigure}{0.45\textwidth}
        \begin{align*}
            \textit{F}&\textit{illBlankString}\triangleq\\ 
            &\lambda \unusedvar.\\
            &\letin{s=\texttt{zeros}_{2r}}{}\\
            &\texttt{let }\lis{\mathit{o},\unusedvar}= 
            \left(
                \begin{aligned}
                    &\texttt{loop }\\
                    &\quad \lambda e.\\
                    &\quad \letin{\lis{s_1,b}=e}{}\\
                    &\quad \texttt{if } b \texttt{ then}\\
                    &\quad\mbox{ } \letin{s_2=\texttt{write}(s_1,1)}{}\\
                    &\quad\mbox{ } \letin{s_3=\texttt{cyclicShiftR}(s_2)}{}\\
                    &\quad\mbox{ } \return{\lis{s_3,\false}}\\
                    &\quad \texttt{else }\\
                    &\quad\mbox{ } \letin{s_2=\texttt{write}(s_1,0)}{}\\
                    &\quad\mbox{ } \letin{s_3=\texttt{cyclicShiftR}(s_2)}{}\\
                    &\quad\mbox{ } \return{\lis{s_3,\true}}\\
                    &2r \texttt{ times from } (\return \lis{s,\false})
                \end{aligned}
            \right)\texttt{ in}\\
            &\return{\mathit{o}}
        \end{align*}
        \caption{Term generating a blank string of length $2r$.}
        \label{subfig:fillBlankString}
    \end{subfigure}
    \mbox{ }
    \begin{subfigure}{0.4\textwidth}
        \begin{align*}
            \textit{F}&\textit{illBlankTape}\triangleq\\ 
            &\lambda \unusedvar.\\
            &\letin{\mathit{lTape}=\fillBlankStringTerm \ast}{}\\
            &\letin{\mathit{rTape}=\fillBlankStringTerm \ast}{}\\
            &\letin{\mathit{symb}=\lis{1,0}}{}\\
            &\return{\lis{\mathit{lTape},\mathit{symb},\mathit{rTape}}}\\[1.5cm]
        \end{align*}
        \caption{Term constructing a blank tape from \\two blank strings.}
        \label{subfig:fillBlankTape}
    \end{subfigure}
    \caption{Implementation of $\lamBLL$ terms for initializing blank tapes. Subfigure~\ref{subfig:fillBlankString} shows the generation of a blank string, and Subfigure~\ref{subfig:fillBlankTape} shows how these strings are combined to form a complete tape.}
    \label{fig:blankTapes}
    \end{minipage}%
}
\end{figure}
 \begin{figure} 
        \makebox[\textwidth][c]{%
        \begin{minipage}{1.2\textwidth}
        \begin{align*}
            \encodeInputStringTerm \triangleq\mbox{ }& \lambda \mathit{input}.\\
            &\letin{\mathit{enc}=\texttt{zeros}_{2p}}{}\\
            &\texttt{let }\lis{\mathit{out},\unusedvar}=
                \left(
                \begin{aligned}
                    &\texttt{loop }\\
                    &\quad \lambda l.\\
                    &\quad \letin{\lis{e,i}=l}{}\\
                    &\quad \letin{b = \texttt{read}(i)}{}\\ 
                    &\quad \letin{t= \texttt{equal}(b,1)}{}\\
                    &\quad \texttt{if } t \texttt{ then}\\
                    &\qquad \letin{e_1=\texttt{write}(e,0)}{}\\
                    &\qquad \letin{e_2=\texttt{cyclicShiftR}(e_1)}{}\\
                    &\qquad \letin{e_3=\texttt{write}(e_2,1)}{}\\
                    &\qquad \letin{e_4=\texttt{cyclicShiftR}(e_3)}{}\\
                    &\qquad \letin{i_1 = \texttt{cyclicShiftR}(i)}{}\\ 
                    &\qquad \return{\lis{e_4,i_1}}\\ 
                    &\quad \texttt{else }\\
                    &\qquad \letin{e_1=\texttt{write}(e,0)}{}\\
                    &\qquad \letin{e_2=\texttt{cyclicShiftR}(e_1)}{}\\
                    &\qquad \letin{e_3=\texttt{write}(e_2,0)}{}\\
                    &\qquad \letin{e_4=\texttt{cyclicShiftR}(e_3)}{}\\
                    &\qquad \letin{i_1 = \texttt{cyclicShiftR}(i)}{}\\ 
                    &\qquad \return{\lis{e_4,i_1}}\\ 
                    &p \texttt{ times from } (\return \lis{\mathit{enc},\mathit{input}})
                \end{aligned}
                \right)\texttt{ in}\\    
                &\return{\mathit{out}}
        \end{align*}
        \caption{Implementation of the $\lamBLL$ term encoding a binary string into its two-bit representation.}
        \label{fig:encodeString}
         \end{minipage}%
}
    \end{figure}
We can the proceed and define a term $\encodeInputStringTerm$ in $\lamBLL$ that encodes a binary input string according to the tape's binary encoding.  
Formally, this term has type
\[
\vdashv[\varnothing] \encodeInputStringTerm : \Str[p] \lollipop \Str[2p],
\]
indicating that it takes a binary string of length $p$ and returns a binary string of length $2p$, where each bit is mapped to its corresponding two-bit code used to represent tape symbols.
Figure~\ref{fig:encodeString} implements $\encodeInputStringTerm$ by iterating over the input string and encoding each bit as its corresponding two-bit code in the encoded string, using cyclic shifts to correctly advance the read position on the input string and the write position on the encoded string.

To initialize the input tape with the binary string, we define the $\lamBLL$ term $\initInputTapeTerm$ of type
\[
\vdashv[\varnothing] \textit{InitInputTape} : \Str[p] \lollipop \typTape
\]
indicating that it takes a binary input string of length $p$ and returns a tape. In this type, $p$ is the polynomial bounding the length of the PTM's input we are encoding. 
The term is in Figure~\ref{fig:initInputTape} and proceeds by initializing the input tape by first encoding the input string into the tape's binary representation using $\encodeInputStringTerm$ and creating a blank tape with $\fillBlankTapeTerm$, then sequentially writing each encoded symbol onto the right half of the tape. The tape head is placed on the first input symbol, while the remaining tape positions stay blank. Formally, the term iterates over the input symbols, writing each onto the tape and cyclically shifting the head to the right. The resulting configuration consists of the left tape, the current symbol under the head (the first input bit), and the right tape containing the rest of the input.

\begin{figure}[h]
        \begin{align*}
            \textit{Init}&\textit{InputTape}\triangleq\\
            &\lambda \mathit{input}.\\
            &\letin{\mathit{enc}=\encodeInputStringTerm \mbox{ }\mathit{input}}{}\\
            &\letin{t=\fillBlankTapeTerm \ast}{}\\
            &\letin{\lis{\mathit{leftTape},\mathit{symb},\mathit{rightTape}}=t}{}\\
            &\texttt{let }\lis{\mathit{trI},i}=
            \left(
            \begin{aligned}
                &\texttt{loop }\\
                &\quad \lambda e.\\
                &\quad \letin{\lis{\mathit{tape},i}=e}{}\\
                &\quad \letin{b=\texttt{read}(i)}{}\\
                &\quad \letin{i_1=\texttt{cyclicShiftR}(i)}{}\\
                &\quad \letin{t_1=\texttt{write}(\mathit{tape},b)}{}\\
                &\quad \letin{t_2=\texttt{cyclicShiftR}(t_1)}{}\\
                &\quad \return{\lis{t_2,i_1}}\\
                &(2p)-2 \texttt{ times from } (\return \lis{\mathit{rightTape},\mathit{enc}})
            \end{aligned}
            \right)\texttt{ in}\\            
            &\letin{s_2=\texttt{read}(i)}{}\\
            &\letin{i_1=\texttt{cyclicShiftR}(i)}{}\\
            &\letin{s_1=\texttt{read}(i_1)}{}\\
            &\return{\lis{\mathit{leftTape},\lis{s_1, s_2},\mathit{trI}}}
        \end{align*}
        \caption{Implementation of $\lamBLL$ term that initializes the input tape.}
        \label{fig:initInputTape} 
    \end{figure}

\begin{sloppypar}
Finally, given an input string, the term $\initInputTapeTerm$ initializes the input tape by positioning the head at the first symbol and filling the rest with blanks. Together, $\encodeInitialStateTerm$, $\initInputTapeTerm$, and $\fillBlankTapeTerm$ are combined into the term $Init$, which returns the complete initial configuration of the PTM, ready for step-wise simulation in $\lamBLL$.
Formally, the term $Init$ is of type
\[
\vdashv[\varnothing] \initTerm: \Str[p] \lollipop \typConf
\]
and it is implemented as follows
\end{sloppypar}
\begin{equation}\label{eq:initTerm}
    \begin{aligned}
        \textit{Init} \triangleq  \mbox{ } &
        \lambda i.\\
        &\letin{\mathit{inputTape}=\initInputTapeTerm\mbox{ }i}{}\\ 
        &\letin{\mathit{workTape}=\fillBlankTapeTerm\ast}{}\\
        &\letin{\mathit{outputTape}=\fillBlankTapeTerm\ast}{}\\
        &\letin{\mathit{state}=\encodeInitialStateTerm \ast}{}\\ 
        &\return{\lis{\mathit{state},\mathit{inputTape},\mathit{workTape},\mathit{outputTape}}}
    \end{aligned}
\end{equation}

\subsubsection{Encoding of Result Extraction.} We define the $\lamBLL$ term $\finalTerm$ to extract the output of a PTM computation. Its type expresses that it takes a full configuration and returns the corresponding output string:
\[
\vdashv[\varnothing] \finalTerm:  \typConf \lollipop \Str[q]
\]
where $q$ is the polynomial bounding the length of the PTM output.

\begin{figure}[h]
    \begin{align*}
        \decodeOutputStringTerm \triangleq\mbox{ }& \lambda \mathit{enc}.\\
        &\letin{\mathit{output}=\texttt{zeros}_{2q}}{}\\
        &\texttt{let }\lis{\unusedvar,\mathit{out}}=
                \left(
                \begin{aligned}
                    &\texttt{loop }\\
                    &\quad \lambda l.\\
                    &\quad \letin{\lis{o,e}=l}{}\\
                    &\quad \letin{b = \texttt{read}(e)}{}\\ 
                    &\quad \letin{t= \texttt{equal}(b,1)}{}\\
                    &\quad \texttt{if } t \texttt{ then}\\
                    &\qquad \letin{o_1=\texttt{write}(o,1)}{}\\
                    &\qquad \letin{o_2=\texttt{cyclicShiftR}(o_1)}{}\\ 
                    &\qquad \letin{e_1 = \texttt{cyclicShiftR}(e)}{}\\ 
                    &\qquad \letin{e_2 = \texttt{cyclicShiftR}(e_1)}{}\\ 
                    &\qquad \return{\lis{e_2,o_2}}\\ 
                    &\quad \texttt{else }\\
                    &\qquad \letin{o_1=\texttt{write}(o,0)}{}\\
                    &\qquad \letin{o_2=\texttt{cyclicShiftR}(o_1)}{}\\ 
                    &\qquad \letin{e_1 = \texttt{cyclicShiftR}(e)}{}\\ 
                    &\qquad \letin{e_2 = \texttt{cyclicShiftR}(e_1)}{}\\
                    &\qquad \return{\lis{e_2,o_2}}\\ 
                    &q \texttt{ times from } (\return \lis{\mathit{enc},\mathit{output}})
                \end{aligned}
                \right)\texttt{ in}\\    
            &\return{\mathit{out}}
    \end{align*}
    \caption{Implementation of the $\lamBLL$ term that decodes a binary string in the tape's two-bit encoding back into the original single-bit representation.}
    \label{fig:decodeOutput}
\end{figure}

First, we define a term $\decodeOutputStringTerm$ in $\lamBLL$ that converts a binary string from the tape's two-bit encoding back to its original form.
Formally, this term has type
\[
\vdashv[\varnothing] \decodeOutputStringTerm : \Str[2q] \lollipop \Str[q],
\]
indicating that it takes a binary string of length $2q$ and returns a binary string of length $q$, where each original bit is reconstructed from its corresponding two-bit code in the tape encoding.
Figure~\ref{fig:decodeOutput} implements $\decodeOutputStringTerm$ by iterating over the output tape and reconstructing the original single-bit string, using cyclic shifts to correctly advance the read position on the encoded string and the write position on the decoded output. 

The term $\finalTerm$ operates by taking a configuration and decomposing it into its components: the current state, the input tape, the work tape, and the output tape. From the output tape, it separates the left and right portions and initializes a blank string of length $2q$ using $\mathtt{zeros}_{2q}$. It then iteratively transfers the symbols from the output tape onto this blank string. At each step, it reads the current symbol from the tape, writes it into the output string, and cyclically shifts both the tape head and the output string head. After $2q$ iterations, the loop constructs the full binary-encoded output string from the tape's symbols, which is then passed to $\decodeOutputStringTerm$ to recover and return the output string. Formally, it is defined as follows

\begin{equation}\label{eq:finalTerm}
    \begin{aligned}
        \finalTerm \triangleq \mbox{ } 
        &\lambda c.\\
        &\letin{\lis{q,i,w,o}=c}{}\\
        &\letin{\lis{\mathit{leftTape},\mathit{symb},\mathit{rightTape}}=o}{}\\
        &\letin{s=\mathit{zeros}_{2q}}{}\\
        &\texttt{let }\lis{\mathit{out},\unusedvar}=
        \left(
            \begin{aligned}
                &\texttt{loop }\\
                &\quad \lambda e.\\
                &\quad \letin{\lis{s_1,t_1}=e}{}\\
                &\quad \letin{b=\texttt{read}(t_1)}{}\\
                &\quad \letin{s_2=\texttt{write}(s_1,b)}{}\\
                &\quad \letin{t_2=\texttt{cyclicShiftR}(t_1)}{}\\
                &\quad \letin{s_3=\texttt{cyclicShiftR}(s_2)}{}\\
                &\quad \return{\lis{s_3, t_2}}\\
                &2q \texttt{ times from } (\return \lis{s,\mathit{leftTape}})
            \end{aligned}
        \right) \texttt{ in}\\
        &\letin{\mathit{dec}=\decodeOutputStringTerm\mbox{ }out}{}\\
        &\return{\mathit{dec}}
    \end{aligned}
\end{equation}

\subsubsection{Putting Everything Together}

Finally, we complete the PTM encoding by defining a $\lamBLL$ term that represents the entire computation of any PTM $\mathcal{M}$ computing a $(p,q)$-function. This embedding is typed as
\[
\vdashv[\varnothing] E_\mathcal{M}: \Str[p] \lollipop \Str[q].
\]
The implementation of $E_\mathcal{M}$ proceeds as follows: given an input, it first constructs the initial configuration of the PTM using $\initTerm$. It then iteratively applies the PTM's step function $\stepPTMTerm$ for $r$ iterations, where $r \in \mathbb{N}_{\geq 1}[i]$ is a polynomial bounding the runtime of the PTM, producing the final configuration. Finally, the output is extracted from the final configuration using $\finalTerm$, which is then returned. Formally, the definition is the following:

\begin{equation}\label{eq:encodingPTM}
    \begin{aligned}
        E_\mathcal{M}\triangleq \mbox{ }&
        \lambda i.\\
        &\letin{\mathit{initialConfig}=\initTerm \mbox{ } i}{}\\
        &\letin{\mathit{finalConfig}=\left(
            \loopp{\stepPTMTerm}{r}{(\return{\mathit{initialConfig}})}
        \right)}{}\\
        &\letin{\mathit{out}=\finalTerm \mbox{ }\mathit{finalConfig}}{}\\
        &\return{\mathit{out}}
    \end{aligned}
\end{equation}

\newcommand{\ism}{\mathbf{ism}}
\newcommand{\PPT}{\mathit{PPT}}
\subsection{Proof of Completeness}

We now prove that the encoding above faithfully simulates the PTM operational semantics.  For a configuration $C$, we write $\enc{C}$ for its $\lamBLL$ encoding.  If $\mathcal{D}\in\mathcal{D}(\mathit{Conf})$, we write $\enc{\mathcal{D}}$ for the push-forward of $\mathcal{D}$ along the encoding, i.e.
\[
\enc{\mathcal{D}}(\enc{C})=\mathcal{D}(C).
\]
Since the terms used in the encoding do not access references, the store is always $\varnothing$ and will be left unchanged by evaluation.

\begin{lemma}[One-step simulation]\label{lemma:OneStepCompleteness}
	Let $C$ be a non-final configuration of a PTM $\mathcal{M}$ and let
	\[
	C\xRightarrow{1}\mathcal{D}.
	\]
	Then, for every security parameter $n$ for which $C$ is a configuration of $\mathcal{M}$ on inputs of size $p(n)$,
	\[
	\ism^{\varnothing;\typConf}_{n}
	(\stepPTMTerm\mbox{ }\enc{C},\varnothing)
	= \enc{\mathcal{D}}.
	\]
\end{lemma}
\begin{proof}
	Write $C=\langle q,i,w,o\rangle$, and let $a$ and $b$ be the symbols currently read on the input and work tapes. By Definition~\ref{def:stepRelationPTM}, one PTM step samples a uniform bit $j\in\{0,1\}$ and applies $\delta_j$ to $(q,a,b)$. Hence, if $C_j$ denotes the configuration obtained by executing the action prescribed by $\delta_j(q,a,b)$, then
	\[
	\mathcal{D}=\tfrac12\delta_{C_0}+\tfrac12\delta_{C_1}.
	\]
	
	Consider now $\stepPTMTerm\mbox{ }\enc{C}$. By Figure~\ref{subfig:step}, after the administrative reductions that unpack $\enc{C}$, the term evaluates $\flipcoin$. By the interpretation of $\flipcoin$ and the definition of $\ism$, the two branches are therefore selected with probability $1/2$ each. In branch $j$, the term $\transitionTerm{j}$ performs the finite case analysis implementing $\delta_j$, and consequently returns the encoding of the action $\delta_j(q,a,b)$. The term $\evalTerm$ then realizes that action: $\writeSymbTerm$ updates the symbol under the work-tape head, while $\moveHeadTerm$ (and hence $\moveLTerm$ or $\moveRTerm$ when required) updates the input and work tapes exactly as prescribed by the corresponding directions; the output action is handled analogously. Thus the evaluation of branch $j$ terminates in $\return{\enc{C_j}}$.
	
	The iterated small-step semantics $\ism$ collects the probabilities of these terminating branches. Therefore
	\[
	\ism^{\varnothing;\typConf}_{n}
	(\stepPTMTerm\mbox{ }\enc{C},\varnothing)
	=\tfrac12\delta_{\enc{C_0}}+\tfrac12\delta_{\enc{C_1}}
	=\enc{\mathcal{D}},
	\]
	as required.
\end{proof}

\begin{lemma}[Multi-step simulation]\label{lemma:StepLemmaCompleteness}
	Let $C$ be a configuration of a PTM $\mathcal{M}$ and suppose that
	\[
	C\xRightarrow{m}\mathcal{D}
	\]
	for some $m\geq 1$. Then
	\[
	\ism^{\varnothing;\typConf}_{n}
	\left(
	\loopp{\stepPTMTerm}{m}{(\return{\enc{C}})},
	\varnothing
	\right)
	=\enc{\mathcal{D}}.
	\]
\end{lemma}
\begin{proof}
	We proceed by induction on $m$.
	
	For $m=1$, unfolding the operational rule for $\loopp{\stepPTMTerm}{1}{(\return{\enc{C}})}$ first reduces the loop to the application of $\stepPTMTerm$ to $\enc{C}$. Lemma~\ref{lemma:OneStepCompleteness} then gives
	\[
	\ism^{\varnothing;\typConf}_{n}
	\left(
	\loopp{\stepPTMTerm}{1}{(\return{\enc{C}})},
	\varnothing
	\right)
	=\enc{\mathcal{D}}.
	\]
	
	Suppose now that the statement holds for $m$, and assume
	$C\xRightarrow{m+1}\mathcal{D}$. Let $\mathcal{D}'$ be the distribution of configurations reached after the first $m$ PTM steps. The $(m+1)$-step distribution is obtained by applying the one-step PTM transition to every configuration in $\mathcal{D}'$ and combining the resulting distributions with the corresponding weights. By the induction hypothesis,
	\[
	\ism^{\varnothing;\typConf}_{n}
	\left(
	\loopp{\stepPTMTerm}{m}{(\return{\enc{C}})},
	\varnothing
	\right)
	=\enc{\mathcal{D}'}.
	\]
	Unfolding once the operational rule for the $(m+1)$-fold loop evaluates one further application of $\stepPTMTerm$ on each value in this distribution. By Lemma~\ref{lemma:OneStepCompleteness}, for every configuration $C'$ in the support of $\mathcal{D}'$, that application yields precisely the encoding of the one-step PTM distribution from $C'$. Since $\ism$ combines probabilistic branches by the distribution monad, the resulting distribution is exactly $\enc{\mathcal{D}}$. This proves the inductive case.
\end{proof}

\Completeness*
\begin{proof}
	We prove the two inclusions separately.
	
	For soundness, let $I=(p,q)$ and let $M$ be a closed $\lamBLL$ value with
	\[
	\vdashv[\varnothing] M:\Str[p]\lollipop\Str[q].
	\]
	By Corollary~\ref{cor:polytimesoundness}, the operational semantics of $M$ can be simulated by a PTM in polynomial time. Hence every $I$-function computed by a $\lamBLL$ term belongs to $\PPT(I)$.
	
	For the converse inclusion, let $f\in\PPT(I)$ and let $\mathcal{M}$ be a PTM computing $f$ in polynomial time. Let $r\in\mathbb{N}_{\geq1}[i]$ be a polynomial time bound for $\mathcal{M}$. Consider the term $E_{\mathcal{M}}$ defined in Equation~\ref{eq:encodingPTM}. By the typings established for $\initTerm$, $\stepPTMTerm$, and $\finalTerm$, together with the typing rule for iteration, we obtain
	\[
	\vdashv[\varnothing]E_{\mathcal{M}}:\Str[p]\lollipop\Str[q].
	\]
	
	Fix $n\geq1$ and an input $s\in\{0,1\}^{p(n)}$, and let $C_{(1^n,s)}$ be the initial configuration of $\mathcal{M}$. By direct inspection of Equation~\ref{eq:initTerm}, evaluating $\initTerm\mbox{ }s$ yields $\enc{C_{(1^n,s)}}$: the state is $q_0$, the input tape contains the two-bit encoding of $s$ with its head on the first symbol, and the work and output tapes are blank. Since $\mathcal{M}$ works in time $r$, there is a distribution $\mathcal{D}$ such that
	\[
	C_{(1^n,s)}\xRightarrow{r(n)}\mathcal{D}.
	\]
	Lemma~\ref{lemma:StepLemmaCompleteness} therefore gives
	\[
	\ism^{\varnothing;\typConf}_{n}
	\left(
	\loopp{\stepPTMTerm}{r}{(\return{\enc{C_{(1^n,s)}}})},
	\varnothing
	\right)
	=\enc{\mathcal{D}}.
	\]
	Finally, by the definition of $\finalTerm$ in Equation~\ref{eq:finalTerm}, each encoded final configuration $\enc{C}$ is mapped to the string written on the output tape of $C$. Consequently, applying $\finalTerm$ to $\enc{\mathcal{D}}$ yields the push-forward of $\mathcal{D}$ to output strings. By the assumption that $\mathcal{M}$ computes $f$, this distribution assigns probability $f_n(s)(t)$ to every $t\in\{0,1\}^{q(n)}$. Hence
	\[
	\ism^{\varnothing;\Str[q]}_{n}
	((E_{\mathcal{M}})_n\,s,\varnothing)
	= (f_n(s),\varnothing),
	\]
	so $E_{\mathcal{M}}$ computes $f$.
	
	Thus $\PPT(I)$ is contained in the class of $I$-functions computed by $\lamBLL$ terms, while the reverse inclusion follows from soundness. Therefore the two classes coincide.
\end{proof}

\end{document}